\documentclass[runningheads]{llncs}

\usepackage{url}

\urldef{\mailsa}\path|{bm, chessy}@upb.de|
\newcommand{\keywords}[1]{\par\addvspace\baselineskip
\noindent\keywordname\enspace\ignorespaces#1}

\usepackage[ngerman,english]{babel}
\usepackage[latin1]{inputenc}
\usepackage[T1]{fontenc}
\usepackage[letterpaper,top=2.54cm,bottom=2.54cm,left=2.54cm,right=2.54cm]{geometry}
\usepackage{amsmath}                   
\usepackage{amssymb}                   
\usepackage{color}                     
\usepackage[noend]{algpseudocode}      
\usepackage[amsmath,thmmarks,thref]{ntheorem} 
\usepackage{graphicx}
\usepackage{stmaryrd}
\usepackage[center]{subfigure}
\usepackage{enumitem}
\usepackage{multirow}

\newcommand{\E}{\ensuremath{\mathbb{E}}}  

\newcommand{\setsymbol}[1]{\mathbb{#1}}
\makeatletter
    \newcommand{\N}{\@ifstar{\setsymbol{N}_0}{\setsymbol{N}}}
\makeatother
\newcommand{\Q}{\setsymbol{Q}}
\newcommand{\R}{\setsymbol{R}}
\newcommand{\Z}{\setsymbol{Z}}
\makeatletter
    \newcommand{\Nupto}{\@ifstar{\Nupto@star}{\Nupto@nostar}}
    \newcommand{\Nupto@star}[1]{[#1]_0}
    \newcommand{\Nupto@nostar}[1]{[#1]}
    \newcommand{\Qpos}{\@ifstar{\Qpos@star}{\Qpos@nostar}}
    \newcommand{\Qpos@star}{\Q_{\geq 0}}
    \newcommand{\Qpos@nostar}{\Q_{> 0}}
    \newcommand{\Rpos}{\@ifstar{\Rpos@star}{\Rpos@nostar}}
    \newcommand{\Rpos@star}{\R_{\geq 0}}
    \newcommand{\Rpos@nostar}{\R_{> 0}}
\makeatother

\DeclareFontFamily{OT1}{mvs}{}
\DeclareFontShape{OT1}{mvs}{m}{n}{<-> fmvr8x}{}

\newcommand{\NP}{$\mathcal{NP}$}
\newcommand{\Pol}{$\mathcal{P}$}

    \newcommand{\pls}{$\mathcal{PLS}$}
    \newcommand{\mc}{\textsc{Max-Cut}}
    \newcommand{\pspace}{$\mathcal{PSPACE}$}

\newlength{\alglabelwidth}

\theoremstyle{plain}
\theoremheaderfont{\normalfont\bfseries}
\theorembodyfont{\itshape}
\theoremseparator{.}
\renewtheorem{theorem}{Theorem}
\renewtheorem{lemma}{Lemma}
\renewtheorem{corollary}{Corollary}
\newtheorem{observation}[theorem]{Observation}
\renewtheorem{claim}[theorem]{Claim}

\theoremstyle{break}
\theorembodyfont{\normalfont}

\theoremstyle{nonumberplain}
\theoremheaderfont{\itshape}
\theorembodyfont{\normalfont}
\theoremsymbol{\squareforqed}
\renewtheorem{proof}{Proof}
{\unskip\nobreak\hskip 1em plus 1fil\nobreak$\Box$
\parfillskip=0pt%
\endtrivlist}

\graphicspath{{Figures/}}

\begin{document}

\title{Settling the Complexity of Local Max-Cut (Almost) Completely \thanks{Partially supported by the German Research Foundation (DFG) Priority Programme 1307 ``Algorithm Engineering''.}}

\titlerunning{Settling the Complexity of Local Max-Cut (Almost) Completely.}
\author{Robert Els\"asser \and Tobias Tscheuschner}
\institute{University of Paderborn \newline Faculty of Computer Science, Electrical Engineering and Mathematics \newline \{elsa, chessy\}@upb.de}

\maketitle

\begin{abstract} We consider the problem of finding a local optimum for \textsc{Max-Cut} with FLIP-neighborhood, in which exactly one node changes the partition. Sch\"affer and Yannakakis (SICOMP, 1991) showed \pls-completeness of this problem on graphs with unbounded degree. On the other side, Poljak (SICOMP, 1995) showed that in cubic graphs every FLIP local search takes $O(n^{2})$ steps, where $n$ is the number of nodes. Due to the huge gap between degree three and unbounded degree, Ackermann, R\"oglin, and V\"ocking (JACM, 2008) asked for the smallest $d \in \N$ for which the local \textsc{Max-Cut} problem with FLIP-neighborhood on graphs with maximum degree $d$ is \pls-complete. In this paper, we prove that the computation of a local optimum on graphs with maximum degree five is \pls-complete. Thus, we solve the problem posed by Ackermann et al.~almost completely by showing that $d$ is either four or five (unless \pls\ $\subseteq$ \Pol). 

On the other side, we also prove that on graphs with degree $O(\log n)$ every FLIP local search has probably polynomial smoothed complexity. Roughly speaking, for any instance, in which the edge weights are perturbated by a (Gaussian) random noise with variance $\sigma^2$, every FLIP local search terminates in time polynomial in $n$ and $\sigma^{-1}$, with probability $1-n^{-\Omega(1)}$. Putting both results together, we may conclude that although local \textsc{Max-Cut} is likely to be hard on graphs with bounded degree, it can be solved in polynomial time for slightly perturbated instances with high probability.
\end{abstract}


\section{Introduction}
For an undirected graph $G=(V,E)$ with weighted edges $w:E \rightarrow \mathbb{N}$ a cut is a 
partition of $V$ into two sets $V_{1},V_{2}$. The weight of the cut is the sum of the weights 
of the edges connecting nodes between $V_1$ and $V_2$. The \textbf{\textsc{Max-Cut}} problem asks for a cut of 
maximum weight. Computing a maximum cut is one of the most famous problems in computer science 
and is known to be \NP-complete even on graphs with maximum degree three \cite{garey1}. For a 
survey of \mc\ including applications see \cite{poljak2}.

A frequently used approach of dealing with hard combinatorial optimization problems is local search. 
In local search, to every solution there is assigned a set of neighbor solutions, i.\ e. a 
neighborhood. The search begins with an initial solution and iteratively moves to better neighbors 
until no better neighbor can be found. For a survey of local search, we refer to \cite{monien2}. To encapsulate many local search problems, Johnson et al.~\cite{johnson1} introduced the complexity class \pls\ (polynomial local search) and initially 
showed \pls-completeness for the \textsc{Circuit-Flip} problem. Sch\"affer and Yannakakis \cite{schaeffer1} 
showed \pls-completeness for many popular local search problems 
including the local \mc\ problem with FLIP-neighborhood -- albeit their reduction builds graphs with linear degree in the worst case. Moreover, they introduced the notion of so 
called \textit{tight} \pls-reductions which preserve not only the existence of instances and 
initial solutions that are exponentially many improving steps away from any local optimum but 
also the \pspace-completeness of the computation of a local optimum reachable by improving steps 
from a given solution. 

In a recent paper Monien and Tscheuschner \cite{monien1} showed the two properties that are preserved by tight \pls-completeness proofs for the local 
\mc\ problem on graphs with maximum degree four. However, their proof did not use a \pls-reduction; 
they left open whether the local \mc\ problem is \pls-complete on graphs with maximum degree four. 
For cubic graphs, Poljak \cite{poljak1} showed that any FLIP-local search takes $O(n^2)$ improving 
steps, where Loebl \cite{loebl1} earlier showed that a local optimum can be found in polynomial 
time using an approach different from local search. Thus, it is unlikely that computing a local 
optimum is \pls-complete on graphs with maximum degree three.

Due to the huge gap between degree three and unbounded degree, 
Ackermann et al.~\cite{ackermann1} asked for the smallest $d$ such that on graphs 
with maximum degree $d$ the 
computation of a local optimum is \pls-complete. 
In this paper, we show that $d$ is either four or five (unless \pls\  $\subseteq$ \Pol), 
and thus solve the above problem almost completely. A related 
problem has been considered by Krentel \cite{krentel1}. He showed \pls-completeness for 
a satisfiability problem with trivalent variables, a clause length of at most four, and maximum 
occurrence of the variables of three. 

Our result has impact on many other problems, since the local \mc\ has been the basis 
for many \pls-reductions in the literature. Some of these reductions directly carry over the 
property of maximum degree five in some sense and result in \pls-completeness of the corresponding 
problem even for very restricted sets of feasible inputs. In particular, \pls-completeness follows 
for the \textsc{Max-2Sat} problem \cite{schaeffer1} with FLIP-neighborhood, in which exactly one variable changes its value, even if every variable occurs at most ten times. \pls-completeness also follows for the 
problem of computing a Nash Equilibrium in \textsc{Congestion Games} (cf.~\cite{fabrikant1},~\cite{ackermann1}) in which each 
strategy contains at most five resources. The 
problem to \textsc{Partition} \cite{schaeffer1} a graph into two equally sized sets of nodes by 
minimizing or maximizing the weight of the cut, where the maximum degree is six and the neighborhood consists of all solutions in which two nodes of different partitions are exchanged, is also \pls-complete. Moreover, our \pls-completeness proof was already helpful showing a complexity result in hedonic games \cite{gairing1}.


In this paper, we also consider the smoothed complexity of any FLIP local search on graphs 
in which the degrees are bounded by $O(\log n)$. This performance measure has been 
introduced by Spielman and Teng in their seminal paper on the smoothed analysis 
of the Simplex algorithm \cite{ST04}\footnote{For this work, Spielman and Teng was awarded the 
G\"odel Prize in 2008.}. Since then, a large number of papers deal with the smoothed 
complexity of different algorithms. In most cases, smoothed analysis 
is used to explain the speed of certain algorithms in practice, which have 
an unsatisfactory running time according to their worst case complexity. 

The smoothed measure of an algorithm on some 
input instance is its expected performance 
over random perturbations of that instance, and the smoothed complexity 
of an algorithm is the maximum smoothed measure over all input instances.   
In the case of an LP, the goal is to 
$\mbox{maximize } z^T x \mbox{ subject to } Ax \leq b $, for given vectors $z$, $b$, and matrix $A$, where
the entries of $A$ are perturbated by Gaussian random variables with mean 
$0$ and variance $\sigma^2$. That is, we add to each entry $a_{i,j}$ some 
value $\max_{i,j} a_{i,j} \cdot y_{i,j}$, where $y_{i,j}$ is
a Gaussian random variable with mean $0$ and standard deviation 
$\sigma$. Spielman and Teng showed that an LP, which is perturbated by some 
random noise as described before, has expected running time polynomial 
in $n$, $m$, and $\sigma$. This result has further been improved by 
Vershynin \cite{Ver06}. The smoothed complexity of 
other linear programming algorithms has been considered in e.g.~\cite{BD02}, and 
quasi-concave minimization was studied in \cite{KN07}. 

Several other algorithms from 
different areas have been analyzed w.\ r.\ t. ~their smoothed complexity (see \cite{ST09} 
for a comprehensive description). 
Two prominent examples of local search algorithms with polynomial smoothed 
complexity are $2$-opt TSP \cite{englert1} and $k$-means \cite{AMR09}. 
We also mention here the 
papers of Beier, R\"oglin, and V\"ocking \cite{BV04,RV07} on the smoothed analysis of integer linear 
programming. They showed that if $\Pi$ is a certain class of integer linear programs,
then $\Pi$ has an algorithm of probably polynomial smoothed complexity\footnote{For 
the definition of probably polynomial smoothed complexity see Section \ref{secsmoothed}.}
iff $\Pi_u \in ZPP$, where $\Pi_u $ is the unary representation of $\Pi$, and 
$ZPP$ denotes the class of decision problems solvable by a randomized algorithm with 
polynomial expected running time that always returns the correct answer. The results 
of \cite{BV04,RV07} imply that e.g.~$0/1$-knapsack, constrained shortest path, and 
constrained minimum weighted matching have probably polynomial smoothed complexity.
Unfortunately, the results of these papers cannot be used to settle the 
smoothed complexity of local \textsc{Max-Cut}.
\paragraph{\textbf{Overview}}

In section \ref{sec:substitute}, we introduce a technique by which we substitute graphs whose nodes of degree greater than five have a certain type -- we will call these nodes \textbf{comparing} -- by graphs of maximum degree five. In particular, we show that certain local optima in the former graphs induce unique local optima in the latter ones. In section \ref{deg5proof} we show an overview of the proof of the \pls-completeness of computing a local optimum of \mc\ on graphs with maximum degree five by reducing from the \pls-complete problem \textsc{CircuitFlip}. In a nutshell, we map instances of \textsc{CircuitFlip} to graphs whose nodes of degree greater than five are comparing. Some parts of the graphs are adjustments of subgraphs of the \pls-completeness proof of \cite{schaeffer1}. Then, using our technique, we show that local optima for these graphs induce local optima in the corresponding instances of \textsc{CircuitFlip}. 

In section \ref{secsmoothed} we show that on graphs with degree $O(\log n)$ local \mc\ has probably polynomial smoothed complexity. To obtain this result, we basically prove that every improving step w.\ r.\ t. the FLIP-neighborhood increases the cut by at least a polynomial value in $n$ and/or $\sigma$, with high probability.

\section{Preliminaries}
A graph $G$ together with a $2$-partition $P$ of $V$ is denoted by $G_{P}$. We let $c_{G_{P}}:V \rightarrow \{0,1\}$ with $c_{G_{P}}(u)=1$ if and only if $u \in V_{1}$ in $G_{P}$. We let $c_{G_{P}}(u)$ be the \textbf{color} of $u$ in $G_{P}$, where $u$ is white if $c_{G_{P}}(u)=0$ and black otherwise. If the considered graph is clear from the context then we also just write $c_{P}(v)$ and if even the partition is clear then we omit the whole subscript. For convenience we treat the colors of the nodes also as truth values, i.\ e. black corresponds to \textit{true} and white to \textit{false}. For a vector $v$ of nodes we let $c(v)$ be the vector of colors induced by $c$. We say that an edge $\{u,v\}$ is \textbf{in the cut} in $P$ if $c_{P}(u) \neq c_{P}(v)$. For a node $u$ we say that $u$ \textbf{flips} if it changes the partition. A node $u$ is \textbf{happy} in $G_{P}$ if a flip of $u$ does not increase the weight of the cut, and \textbf{unhappy} otherwise. Since we consider weighted graphs, we also say that a flip increases the cut if it increases the weight of the cut. A partition $P$ is a \textbf{local optimum} if all nodes in $G_{P}$ are happy. 

A local search problem $\Pi$ consists of a set of instances $\mathcal{I}$, a set of feasible solutions $\mathcal{F}(I)$ for every instance $I \in \mathcal{I}$, and an objective function $f:\mathcal{F}(I) \rightarrow \Z$. In addition, every solution $s \in \mathcal{F}(I)$ has a neighborhood $\mathcal{N}(s,I) \subseteq \mathcal{F}(I)$. For an instance $I \in \mathcal{I}$, the problem is to find a solution $s \in \mathcal{F}(I)$ such that for all $s' \in \mathcal{N}(s,I)$ solution $s'$ does not have a greater value than $s$ with respect to $f$ in case of maximization and not a lower value in case of minimization.
 	
A local search problem $\Pi$ is in the class $\boldsymbol{\mathcal{PLS}}$ \cite{johnson1} if the following three polynomial time algorithms exist: algorithm A computes for every instance $I \in \mathcal{I}$ a feasible solution $s \in \mathcal{F}(I)$, algorithm B computes for every $I \in \mathcal{I}$ and $s \in \mathcal{F}(I)$ the value $f(s)$, and algorithm C returns for every $I \in \mathcal{I}$ and $s \in \mathcal{F}(I)$ a better neighbor solution $s' \in \mathcal{N}(s,I)$ if there is one and ``locally optimal'' otherwise. A problem $\Pi \in $ \pls\ is $\boldsymbol{\mathcal{PLS}}$\textbf{-reducible} to a problem $\Pi' \in $ \pls\ if there are the following polynomial time computable functions $\Phi$ and $\Psi$. The function $\Phi$ maps instances $I$ of $\Pi$ to instances of $\Pi'$ and $\Psi$ maps pairs $(s,I)$, where $s$ is a solution of $\Phi(I)$, to solutions of $I$, such that for all instances $I$ of $\Pi$ and local optima $s^{*}$ of $\Phi(I)$ the solution $\Psi(s^{*},I)$ is a local optimum of $I$. Finally, a problem $\Pi \in $ \pls\ is $\boldsymbol{\mathcal{PLS}}$\textbf{-complete} if every problem in \pls\ is \pls-reducible to $\Pi.$

In our technique, as well as in the \pls-completeness proof, we make use of a result of Monien and Tscheuschner \cite{monien1}. They showed a property for a set of graphs containing two certain types of nodes of degree four. 
Since we do not need their types in this paper, we omit the restrictions on the nodes and use the following weaker proposition.
\begin{lemma}[\cite{monien1}]
\label{Gf} Let $C_{f}$ be a boolean circuit with $N$ gates which computes a function $f:\{0,1\}^{n} \rightarrow \{0,1\}^{m}$. Then, using $O(log N)$ space, one can compute a graph $G_{f}=(V_{f},E_{f})$ with maximum degree four containing nodes $s_{1}, \ldots ,s_{n},$ $t_{1}, \ldots ,t_{m}$ $\in V_{f}$ of degree one such that for the vectors $s:=(s_1, \ldots ,s_n), t:=(t_1, \ldots ,t_n)$ we have $f(c_{P}(s))=c_{P}(t)$ in every local optimum $P$ of $G_{f}$.
\end{lemma}

\begin{definition}
For a polynomial time computable function $f$ we say that $G_{f}=(V_{f},E_{f})$ as constructed in Lemma \ref{Gf} is the graph that \textit{\textbf{looks}} at the input nodes $s_i \in V_{f}$ and \textit{\textbf{biases}} the output nodes $t_i \in V_{f}$ to take the colors induced by $f$. 
\end{definition}

\paragraph{\textbf{Usage of Lemma \ref{Gf}}} Notice first that $G_{f}$ can be constructed in logarithmic space and thus polynomial time for any polynomial time computable function $f$. In the rest of the paper we use the graph $G_{f}$ for several functions $f$ and we will scale the weights of its edges. Then, the edges of $G_{f}$ give incentives of appropriate weight to certain nodes of those graphs to which we add $G_{f}$. The incentives bias the nodes to take the colors induced by $f$. We already point out that for any node $v$ we will introduce at most one subgraph that biases $v$. Moreover, the unique edge $e=\{u,v\}$ incident to a biased node $v$ that is an edge of the subgraph that biases $v$ will in many cases have the lowest weight among the edges incident to $v$. In particular, the weight of $e$ will then be chosen small enough such that the color of $v$, in local optima, depends on the color of $u$ if and only if $v$ is indifferent with respect to the colors of the other nodes adjacent to $v$. Note that in local optima the node $u$ has the opposite color as the color to which $v$ is biased according to $f$. 

\section{Substituting certain nodes of unbounded degree}\label{sec:substitute}


\begin{definition}
Let $G=(V,E)$ be a graph. A node $v \in V$ is called \textbf{comparing} if there is an $m \in \N$ such that
\begin{enumerate}[label=(\roman*)]
		\item $v$ is adjacent to exactly $2m+1$ nodes $u_1^1,u_1^2,u_2^1,u_2^2,\ldots,u_m^1,u_m^2, u \in V \setminus \{v\}$ with edge weights $a_1,\ldots,a_m,\delta,$ as shown in Figure \ref{nodeToSubstitute},
		\item $u$ is a node of a subgraph $G'=(V',E')$ of $G$ that looks at a subset of $V  \setminus \{u,v\}$ and biases $v$,
		\item $a_i \ge 2a_{i+1}$ for all $1 \le i < m$ and $a_{m} \ge 2 \delta$.
\end{enumerate}
The subgraph $G'$ is called the \textbf{biaser} of $v$. For $u_{i}^{j}$ with $1 \le i \le m$, $1 \le j \le 2$ we call the node $u_{i}^{k}$ with $1 \le k \le 2$ and $k \ne j$ adjacent to $v$ via the unique edge with the same weight as $\{u_{i}^{j}, v\}$ the \textbf{counterpart} of $u_{i}^{j}$ with respect to $v$.
\end{definition}

	\begin{figure}[htb]
		\centering
			\includegraphics[width=0.43\textwidth]{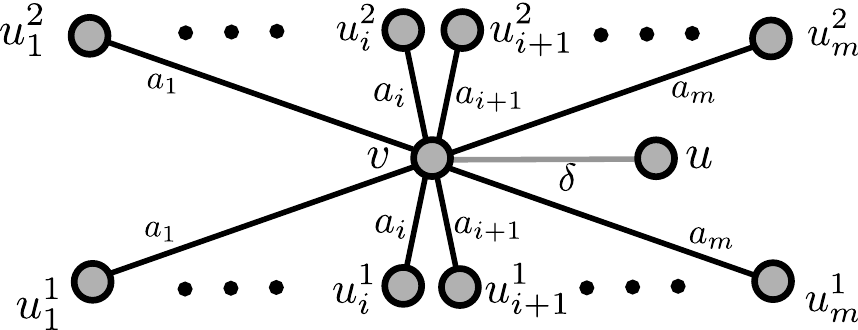}
			\caption{Node $v$ is a comparing node.}
			\label{nodeToSubstitute}
	\end{figure}

The name of the comparing node stems from its behaviour in local optima. If we treat the colors of the neighbors $u_{1}^1,\ldots,u_m^1$ of $v$ as a binary number $a$, with $u_{1}^1$ being the most significant bit, and the colors of $u_{1}^2,\ldots,u_m^2$ as the bitwise complement of a binary number $b$ then, in a local optimum, the comparing node $v$ is white if $a > b$, it is black if $a < b$, and if $a=b$ then $v$ has the color to which it is biased by its biaser. In this way, the color of $v$ ``compares'' $a$ and $b$ in local optima.

	\begin{figure}[htb]
		\centering
			\includegraphics[width=0.8\textwidth]{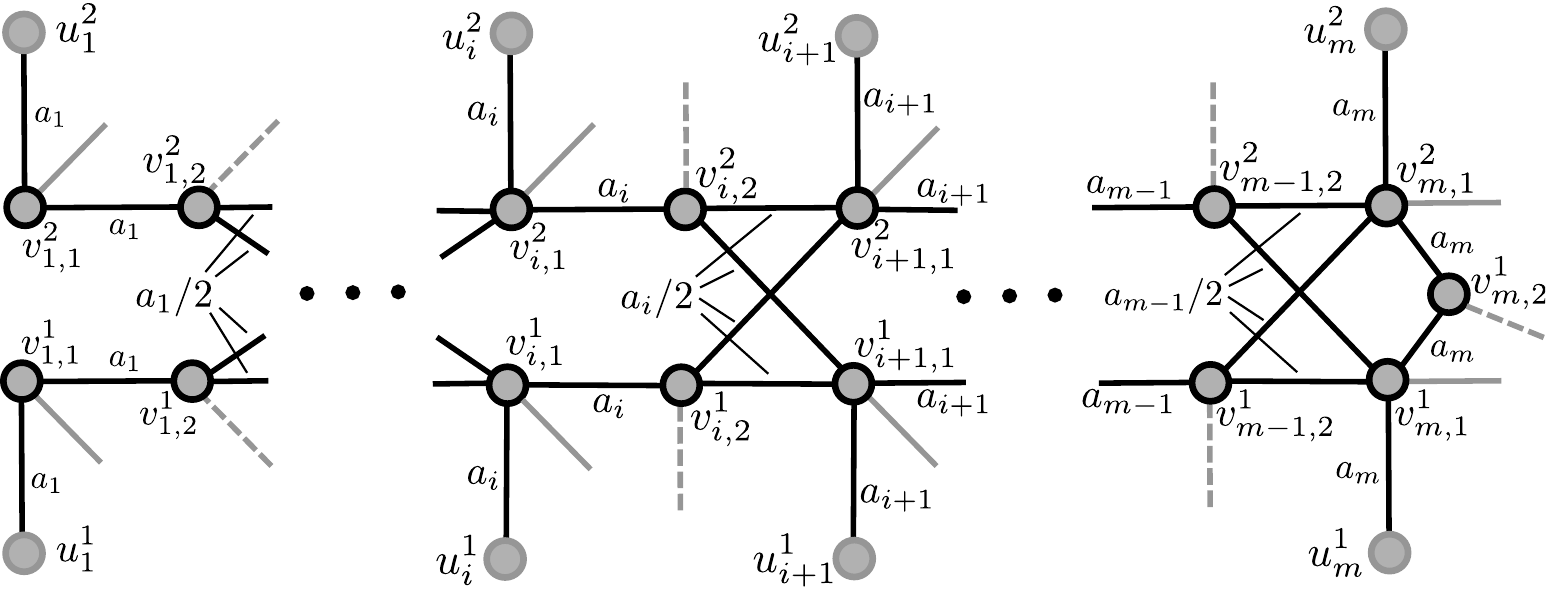}
			\caption{The gadget that substitutes a comparing node $v$.}
			\label{nodeSubstitute}
	\end{figure}	

In the following, we let $G=(V,E)$ be a graph and $v \in V$ be a comparing node with adjacent nodes and incident edges as in Figure \ref{nodeToSubstitute}. We say that we \textbf{degrade} $v$ if we remove $v$ and its incident edges and add the following nodes and edges. We introduce nodes $v_{i,j}^{k}$ for $1 \leq i < m, 1 \leq j \leq 2, 1 \leq k \leq 2$, nodes $v_{m,1}^{k}$ for $1 \leq k \leq 2$, and $v_{m,2}^{1}$ with edges and weights as depicted in Figure \ref{nodeSubstitute} -- the nodes $u_{i}^{j}$ in Figure \ref{nodeSubstitute} have gray circumcircles to indicate that they, in contrast to the other nodes, also occur in $G$. Furthermore, we add a subgraph $G''$ that looks at $u$ and biases all nodes $v_{i,1}^{k}$ to the opposite of the color of $u$ (this is illustrated by short gray edges in Figure \ref{nodeSubstitute}) and the nodes $v_{i,2}^{k}$ to the color of $u$ (short gray dashed edges). The weights of the edges of $G''$ are scaled such that each of them is strictly smaller than $\delta$. Note that due to the scaling the color of the unique node of $G''$ adjacent to $u$ does not affect the happiness of $u$ -- node $u$ is therefore not depicted in Figure \ref{nodeSubstitute} anymore. We let \textbf{\textit{G(G,v)}} be the graph obtained from $G$ by degrading $v$ and we call $v$ \textbf{weakly indifferent} in a partition $P$ if $c_P(u_i^1) \ne c_P(u_i^2)$ for all $1 \le i \le m$. If $v$ is not weakly indifferent then we call the two nodes $u^{1}_{i},u^{2}_{i}$ adjacent to $v$ via the edges with highest weight for which $c_{P}(u^{1}_{i})=c_{P}(u^{2}_{i})$ the \textbf{decisive neighbors} of $v$ in $P$. We let $V_{com} \subseteq V$ be the set of comparing nodes of $V$, and for a partition $P$ of the nodes of $G(G,v)$ we let $\boldsymbol{col_P}: V_{com} \rightarrow \{0,1\}$ be the partial function defined by

$$ col_P(v)=\left\{\begin{array}{cl} 0, & \mbox{if for all } i,j: c_{P}(v_{i,1}^{j}) = 0  \mbox{ and } c_{P}(v_{i,2}^{j}) = 1, \\ 1, & \mbox{if for all } i,j: c_{P}(v_{i,1}^{j}) = 1  \mbox{ and } c_{P}(v_{i,2}^{j}) = 0. \end{array}\right. $$

We say that a comparing node $v$ has the color $\kappa \in \{0,1\}$ in a partition $P$ if $col_P(v)=\kappa$.

\begin{theorem}\label{substitute} Let $G=(V,E)$ be a graph, $v \in V$ a comparing node, its adjacent nodes and incident edges as in Figure \ref{nodeToSubstitute}, $P$ be a local optimum of $G$ such that in $P$ the biaser of $v$ biases $v$ to $c_P(v)$, i.\ e. $c_{P}(u) \ne c_{P}(v)$. Let $P'$ be a partition of the nodes of $G(G,v)$ such that $c_{P}(w)=c_{P'}(w)$ for all $w \in V \setminus \{v\}$. Then, $P'$ is a local optimum if and only if $c_{P}(v) = col_{P'}(v)$.
\end{theorem}

Note the restriction that in the local optimum $P$ the biaser of $v$ biases $v$ to the color that $v$ in fact has in $P$ and \textit{not} to the opposite. In the \pls-completeness proof in section \ref{deg5proof} the biaser of any comparing node $v$ is designed to bias $v$ to the color that $v$ has in a local optimum $P$ due the colors of its neighbors. Then, we can use Theorem \ref{substitute} to argue about $col_{P}(v)$ in $G(G,v)$.

\begin{proof} Let $\kappa \in \{0,1\}$ be the color to which $v$ is biased by its biaser in $P$, i.\ e. $\kappa := c_{P}(v).$ For all $i,j$ we call the color of $v_{i,1}^{j}$ \textbf{correct} if $c_{P'}(v_{i,1}^{j}) = \kappa$ and we call the color of $v_{i,2}^{j}$ correct if $c_{P'}(v_{i,2}^{j}) = \overline{\kappa}$. Moreover, we call $v_{i,j}^{k}$ correct for any $i,j,k$ if it has its correct color. 

``$\Rightarrow$'': Let $P'$ be a local optimum. Note that each node $v_{i,j}^{k}$ is biased by an edge with weight lower than $\delta$ to its correct color. Therefore, to show that it is correct in the local optimum $P'$, it suffices to show that it gains at least half of the sum of weights of the incident edges with weight greater than $\delta$ if it is correct. We prove the Theorem by means of the following Lemmas which are each proven via straightforward inductive arguments.

\begin{lemma}\label{amrand}
Let $q \le m$ and $c_{P}(u_{i}^{1}) = \overline{\kappa}$ for all $i \le q$. Then, $v_{i,1}^{1}$ and $v_{i,2}^{1}$ are correct for all $i \le q$.
\end{lemma}

\begin{proof}
We prove the claim by induction on $i$. Due to $c_{P}(u_{1}^{1}) = \overline{\kappa}$ we get the correctness of $v_{1,1}^{1}$. For each $i \le q$ the correctness of $v_{i,1}^{1}$ implies the correctness of $v_{i,2}^{1}$. Moreover, for each $i < q$ the correctness of $v_{i,2}^{1}$ together with $c_{P}(u_{i+1}^{1}) = \overline{\kappa}$ implies the correctness of $v_{i+1,1}^{1}$.
\end{proof}

\begin{lemma}\label{nachoben}
Let $q \le m$, node $v_{i,1}^{1}$ and $v_{i,2}^{1}$ be correct for all $i \le q$, and $v_{q,1}^{2}$ be correct. Then, $v_{i,1}^{2}$ and $v_{i,2}^{2}$ are correct for all $i < q.$
\end{lemma}

\begin{proof}
We prove the claim by induction on $i$. Node $v_{q,1}^{2}$ and $v_{q,1}^{1}$ are correct by assumption. For each $i < q$ node $v_{i,2}^{2}$ is correct if $v_{i+1,1}^{2}$ is correct since $v_{i+1,1}^{1}$ is correct by assumption. Moreover, for each $1 < i < q$ node $v_{i,1}^{2}$ is correct if $v_{i,2}^{2}$ is correct since $v_{i-1,2}^{1}$ is correct by assumption. Finally, node $v_{1,1}^{2}$ is correct if $v_{1,2}^{2}$ is correct.
\end{proof}

\begin{lemma}\label{nachunten}
Let $q \le m$. If $v_{q,1}^{1}$ and $v_{q,1}^{2}$ are correct then $v_{i,j}^{k}$ is correct for any $j,k,$ and $q \leq i \leq m.$
\end{lemma}

\begin{proof}
If $q=m$ then the correctness of $v_{m,1}^{1}$ implies the correctness of $v_{m,2}^{1}$. The case $q < m$ is done by induction on $i$. Node $v_{q,1}^{1}$ and $v_{q,1}^{2}$ are correct by assumption. Assume that $v_{i,1}^{1}$ and $v_{i,1}^{2}$ are correct for an arbitrary $q \le i < m.$ Then, the nodes $v_{i,2}^{1}$ and $v_{i,2}^{2}$ are correct whereafter the correctness of $v_{i+1,1}^{1}$ and $v_{i+1,1}^{2}$ follows. Finally, the correctness of $v_{m,1}^{1}$ implies the correctness of $v_{m,2}^{1}.$
\end{proof}


We first consider the case that $v$ is weakly indifferent. Then, for each $i$ at least one of the nodes $u_{i}^{1}$ and $u_{i}^{2}$ has the color $\kappa$. Due to the symmetry between the nodes $v_{i,j}^{1}$ and $v_{i,j}^{2}$ we may assume w.\ l.\ o.\ g. that $c_{P}(u_{i}^{1}) = \overline{\kappa}$ for all $i.$ Then, Lemma \ref{amrand} implies that $v_{i,1}^{1}$ and $v_{i,2}^{1}$ are correct for all $i$. Then, the correctness of $v_{m,2}^{1}$ and $v_{m-1,2}^{1}$ together imply the correctness of $v_{m,1}^{2}$. Then, Lemma \ref{nachoben} implies the correctness of $v_{i,1}^{2}$ and $v_{i,2}^{2}$ for all $i < m.$

Now assume that $v$ is not weakly indifferent and let $u_{q}^{1}$ and $u_{q}^{2}$ be the decisive neighbors of $v$. As in the previous case we assume w.\ l.\ o.\ g. that $c_{P}(u_{i}^{1}) = \overline{\kappa}$ for all $i \le q.$ Then, due to Lemma \ref{amrand} node $v_{i,1}^{1}$ and $v_{i,2}^{1}$ are correct for all $i \le q$. If $q=1$ then $c(u_{1}^{2})=\overline{\kappa}$ implies the correctness of $v_{1,1}^{2}$ -- recall that by assumption $v$ is biased to the opposite color of the color of the decisive nodes. On the other hand, if $q > 1$ then the correctness of $v_{q-1,2}^{1}$ and $c(u_{q}^{2})=\overline{\kappa}$ together imply the correctness of $v_{q,1}^{2}$. Then, Lemma \ref{nachoben} implies the correctness of $v_{i,1}^{2}$ and $v_{i,2}^{2}$ for all $i < q.$ Finally, Lemma \ref{nachunten} implies the correctness of $v_{i,j}^{k}$ for all $j,k,$ and $q \leq i \leq m.$ 

``$\Leftarrow$'': Assume, that every node $v_{i,j}^{k}$ is correct. As we have seen in ``$\Rightarrow$'' $v_{i,j}^{k}$ is happy then. Moreover, each $u_{i}^{j}$ is also happy since its neighbors have the same colors as in the local optimum $P$ --  recall that if $v_{i,1}^{j}$ is correct it has the same color in $P'$ as $v$ in $P$. The colors of the remaining nodes are unchanged. Therefore, $P'$ is a local optimum. This finishes the proof of Theorem  \ref{substitute}.
\end{proof}

\section{Proof of \pls-Completeness}\label{deg5proof}

Our reduction bases on the following \pls-complete problem \textsc{CircuitFlip} (in \cite{johnson1} it is called \textsc{Flip}, which we avoid in this paper since the neighborhood of \mc\ has the same name).

\begin{definition}[\cite{johnson1}]
An instance of \textbf{\textsc{CircuitFlip}} is a boolean circuit $C$ with $n$ input bits and $m$ output bits. A feasible solution of \textsc{CircuitFlip} is a vector $v \in \{0,1\}^{n}$ of input bits for $C$ and the value of a solution is the output of $C$ treated as a binary number. Two solutions are neighbors if they differ in exactly one bit. The objective is to maximize the output of $C$.
\end{definition}
\begin{theorem}\label{plsred}
The problem of computing a local optimum of the \mc\ problem on graphs with maximum degree five is \pls-complete.
\end{theorem}

\begin{proof}
We reduce from the \pls-complete problem \textsc{Circuitflip}. Let $C$ be an instance of \textsc{Circuitflip} with input variables $X_{1},\ldots,X_{n}$, outputs $C_{1},\ldots,C_{m}$, and gates $G_{N},\ldots,G_{1}$. W.\ l.\ o.\ g. we make the following assumptions. Each input variable occurs exactly once in exactly one gate. All gates are NOR-gates with a fanin of 2 and are topologically sorted such that $i > j$ if $G_{i}$ is an input of $G_{j}$. For the sake of simplicity, we denote $G_i$ also as the output of gate $G_i$. The two inputs of a gate $G_{i}$ are denoted by $I_{1}(G_{i})$ and $I_{2}(G_{i})$, i.\ e. a gate $G_{i}$ computes correctly if and only if $G_{i}=\neg (I_{1}(G_{i}) \vee I_{2}(G_{i}))$. For no gate $G_i$ we have $I_{1}(G_{i})=I_{2}(G_{i})$. The gates $G_{1},\ldots,G_{m}$ are the output of $C$ where $G_{m}$ is the most significant bit and $G_{m+1},\ldots,G_{2m}$ compute the corresponding negations of the output bits. The gates $G_{2m+1},\ldots,G_{2m+n}$ and $G_{2m+n+1},\ldots,G_{2m+2n}$ return the same better neighbor solution if there is one and return $X_{1},\ldots,X_{n}$ otherwise. Finally, let $C(x)$ be the output of $C$ on input $x \in \{0,1\}^{n}$ and $w(x)$ be the better neighbor of $x$ computed by $C$ on input $x$ and assume w.\ l.\ o.\ g. $N>20$ and $m \geq n>4$. 

\textit{The proof in a nutshell:} From $C$ we construct a graph $G_{C}$ consisting of two isomorphic subgraphs $G_{C}^{0},G_{C}^{1}$ representing copies of $C$ -- the overall structure of our proof is inspired by \cite{krentel1}. For each gate $G_{i}$ in $C$ there is a subgraph $S_{i}^{\kappa}$ for $\kappa \in \{0,1\}$ in $G_{C}$. The subgraphs $S_{i}^{\kappa}$ are taken from \cite{schaeffer1} and adjusted such that they have maximum degree five without changing local optima. In particular, each $S_{i}^{\kappa}$ contains a comparing node $g_{i}^{\kappa}$ whose color represents the output of $G_{i}$. To maintain a maximum degree of five we assume that $g_{i}^{\kappa}$ is degraded in $G_{C}$ and argue via Theorem \ref{substitute} about its color in local optima. Then, the colors of the nodes of $S_{i}^{\kappa}$, in local optima, either behave as a NOR-gate or have a \textbf{reset} state, i.\ e. a state in which each input node of $S_{i}^{\kappa}$ is indifferent w.\ r.\ t. its neighbors in $S_{i}^{\kappa}$. For each $\kappa \in \{0,1\}$ we have a subgraph $T^{\kappa}$ that looks at $g_{i}^{\kappa}$ for $2m+1 \le i \le 2m+n$, i.\ e. at the improving solution, and biases each input node of $G_{C}^{\overline{\kappa}}$ to the color of its corresponding $g_{i}^{\kappa}$. Finally, we have a subgraph that looks at the input nodes of $G_{C}^{0},G_{C}^{1}$, decides whose input results in a greater output w.\ r.\ t. $C$ -- this subgraph is called \textbf{winner} as opposed to the \textbf{loser} which is the other subgraph -- and biases the subgraphs $S_{i}^{\kappa}$ of the winner to behave like NOR-gates and the subgraphs of the loser to take the reset state. Then, we show that the colors of the subgraphs $S_{i}^{\kappa}$ of the winner in fact reflect the correct outputs w.\ r.\ t their inputs and that the input nodes of the loser in fact are indifferent w.\ r.\ t. their neighbors in the subgraphs $S_{i}^{\kappa}$. Then, due to the bias of $T^{\kappa}$, the input nodes of the loser take the colors of the improving neighbor computed by the winner whereafter the loser becomes the new winner. Hence, the improving solutions switch back and forth between the two copies until the colors of the input nodes of both copies are local optima and the copies return their input as improving solution. Then, the colors of the input nodes induce a local optimum of $C$.

Before turning into the details we introduce some notations w.\ r.\ t. $G_{C}$. We let $x_{i}^{\kappa}$ be the input nodes of $G_C^{\kappa}$, $w_{i,1}^{\kappa}:=g_{2m+i}^{\kappa},w_{i,2}^{\kappa}:=g_{2m+n+i}^{\kappa}$ for $1 \le i \le n$, and $\hat{g}_{i}^{\kappa}:=g_{m+i}^{\kappa}$ for $1 \leq i \leq m$. Each subgraph $G_C^{\kappa}$ also contains nodes $y_{i}^{\kappa}$, $z_{i}^{\kappa}$ for $0 \le i \le 2N+1$ and $\lambda_{i}^{\kappa}$ for $1 \le i \le n$ which induce vectors $y^{\kappa}$, $z^{\kappa}$, and $\lambda^{\kappa}$. Moreover, we let $x^{\kappa}$ be the vector of nodes induced by $x_{i}^{\kappa}$ for $1 \le i \le n$.


We will introduce the nodes and edges of $G_{C}$ via so called components. A component of $G_{C}$ is a tupel $(V'_{C},E'_{C})$ with $V'_{C} \subseteq V_{C}$ and $E'_{C} \subseteq E_{C}$. The components of $G_{C}$ have fourteen types: type 1 up to type 14, where we say that the nodes, edges, and weights of the edges of the components have the same types as their corresponding components. We will explicitly state weights for the edges of type 2 up to 7. However, the weights of these components are only stated to indicate the relations between edge weights of the same type. The only edge weights that interleave between two different types are those of type $3$ and $4$. The edges of type $3$ and $4$ are scaled by the same number. For all other types we assume that their weights are scaled such that the weight of an edge of a given type is greater than four times the sum of the weights of the edges of higher types combined. Note that for these types a lower type implies a higher edge weight. To distinguish between the meaning of the explicitly stated edge weights and the final edge weights, i.\ e. the weights resulting by the scale, we will speak of the explicitly stated weights of \textbf{relative edge weights}.

The components of some types are introduced via drawings. In the drawings, the thick black edges and the nodes with black circumcircles are nodes counted among the components of the introduced type. Gray edges and nodes with gray circumcircles are of a different type than the component introduced in the corresponding drawing and are only (re-)drawn to simplify the verification of the proofs for the reader -- in particular the condition that each node is of maximum degree five. If for a gray edge there is no explicit relative weight given then the edge is among the types $8-14.$ If a gray edge is dotted then it is of higher type than the non-dotted gray edges of the same drawing. If a node has a black or a white filling then it is of type 1. These nodes are also (re-)drawn in components of type higher than 1.

\textbf{Type 1} is to provide the constants $0$ and $1$ for the components of higher type. It contains nodes $s,t$ which are connected by an edge with a weight that is greater than the sum of all other edges in $E_{C}$. Assume w.\ l.\ o.\ g. $c(s)=0$ and let $S$ and $T$ be the sets of nodes representing the constants $0$ and $1$. Type 1 looks at $s$ and biases the nodes of $S$ to the color of $s$ and the nodes of $T$ to the opposite. In the following we assume for each constant introduced in components of higher types there is a separate node in the sets $S,T$.

\textbf{Type 2} contains the nodes  $d^{0}, d^{1}, u^{0}, u^{1}$ -- we will see later that $d^{0}$ and $d^{1}$ are comparing nodes -- with edges and relative weights as depicted in Figure \ref{di}. The purpose of these edges is -- together with the edges of type 9 and 10 -- to guarantee that $d^{0}$ and $d^{1}$ are not both black in local optima. The nodes $d^{0}$ and $d^{1}$ are adjacent to many nodes of higher type, and have a degree greater than five.

	\begin{figure}[htb]
		\centering
			\includegraphics[width=0.27\textwidth]{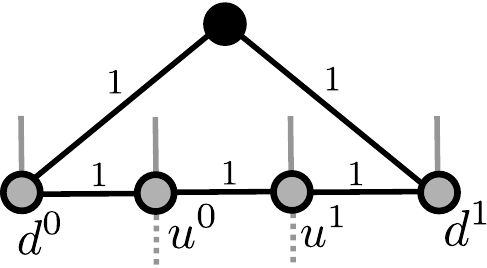}
			\caption{The component of type 2}
			\label{di}
	\end{figure}	
	
The components of type $3$ to $7$ are to represent the two subgraphs $G_C^{0}$ and $G_C^{1}$. The components are very similar to certain clauses of \cite{schaeffer1}. There are three differences between our components and their clauses. First, we omit some nodes and edges to obtain a maximum degree of five for all nodes different from $g_{i}^{\kappa}, I_{1}(g_{i}^{\kappa}),$ and $I_{2}(g_{i}^{\kappa})$. Second, we use different edge weights. However, the weights are manipulated in a way such that the happiness of each node for given colors of the corresponding adjacent nodes is the same as in \cite{schaeffer1}. Third, we add nodes that we bias and to which we look at. Their purpose is to derive the color that a comparing node $g_{i}^{\kappa}$ would have if it was a single node. This color is used to bias $g_{i}^{\kappa}$ such that Theorem \ref{substitute} implies either $col(g_{i}^{\kappa})=0$ or $col(g_{i}^{\kappa})=1$.

\textbf{Type 3} consists of subgraphs $S_{i}^{\kappa}$ which are to represent the gates $G_{i}$ of $C$. For gates whose inputs are not inputs of $G_C^{\kappa}$ they are depicted in Figure \ref{fig:gates}. Together with $d^{0}$ and  $d^{1}$, the nodes $g_{i}^{\kappa}$ (and $I_{k}(g_{j}^{\kappa})$ respectively) are the only nodes which have a degree greater than five -- we will see later that they are also comparing. For each gate $G_{i}^{\kappa}$ whose inputs are inputs of $G_C^{\kappa}$ we take the same components as for those gates whose inputs are not inputs of $G_C^{\kappa}$ but make the following adjustment. We omit the edges $\{I_{1}(g_{i}^{\kappa}),0\}$ and $\{I_{2}(g_{i}^{\kappa}),1\}$ and subtract their relative weights from the edges $\{I_{1}(g_{i}^{\kappa}),1\}$ and $\{I_{2}(g_{i}^{\kappa}),0\}$ respectively, i.\ e.  their relative weights are $2^{10i+7}-2^{10i-5}$ and $2^{10i-5}-2^{10i-1}$. Note that the adjustment does not change the happiness of the nodes $I_{1}(g_{i}^{\kappa})$ and $I_{2}(g_{i}^{\kappa})$ for any given colors of themselves and their neighbors. We call the edges $\{g_{i}^{\kappa},u_{i,j}^{\kappa}\}$ for $j \in \{2,3,6,7,10,11\}$ \textbf{corresponding to} $g_{i}^{\kappa}.$

			\begin{figure}[h!t!b!]
				\centering
				\subfigure{%
					\label{fig:I1(gi)}%
				\includegraphics[width=0.55\textwidth]{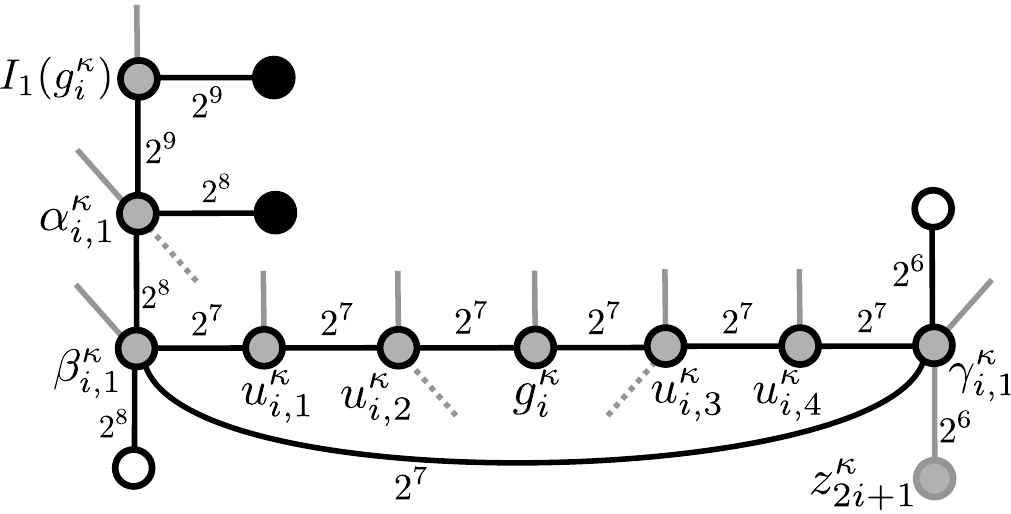}}
				\subfigure{%
				\qquad
					\includegraphics[width=0.55\textwidth]{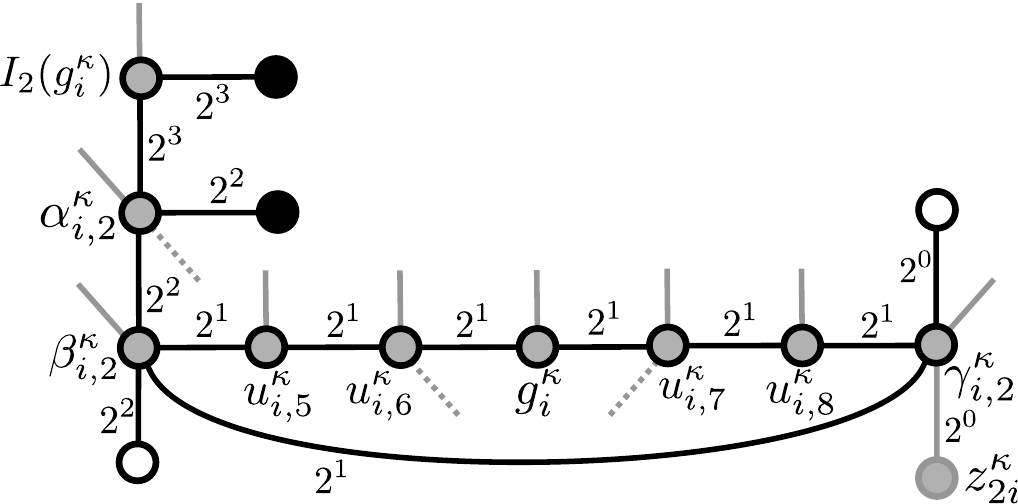}%
					\label{fig:I3(gi)}\qquad}
				\subfigure{%
				\qquad
					\includegraphics[width=0.55\textwidth]{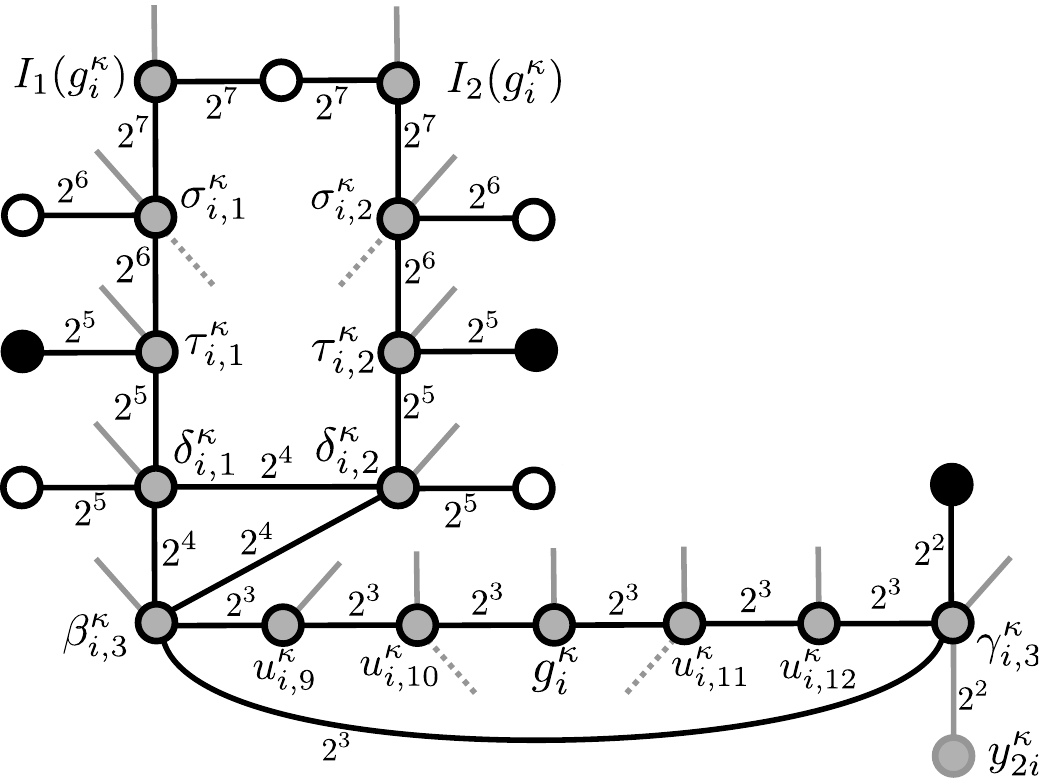}%
					\label{fig:I2(gi)}\qquad}
				\caption{The components of type 3; extra factor for relative edge weights: $2^{10i-2}$}
				\label{fig:gates}
			\end{figure}			
			
\textbf{Type 4} (Figure \ref{yz}) checks whether the outputs of the gates represented by the components of type 3 are correct and gives incentives to nodes of other components depending on the result. As in \cite{schaeffer1} we say that the \textbf{natural value} of the nodes $y_{i}^{\kappa}$ is $1$ and the natural value of the nodes $z_{i}^{\kappa}$ is $0$. The nodes $y_{N+1}^{\kappa},z_{N+1}^{\kappa},\ldots,y_{2}^{\kappa},z_{2}^{\kappa}$ check the correct computation of the corresponding gates and give incentives to their corresponding gates depending on whether the previous gates are correct. The nodes $y_{1}^{\kappa},z_{1}^{\kappa},y_{0}^{\kappa},z_{0}^{\kappa}$ are to give incentives to $d^0,d^1$ depending on whether all gates are correct. Recall that the weights of the edges of type $4$ are the only weights that interleave with weights of edges of a higher type, namely with those of type $3$.
	
			\begin{figure}[htb]
				\centering
				\subfigure{%
					\label{fig:yizi}%
				\includegraphics[width=0.92\textwidth]{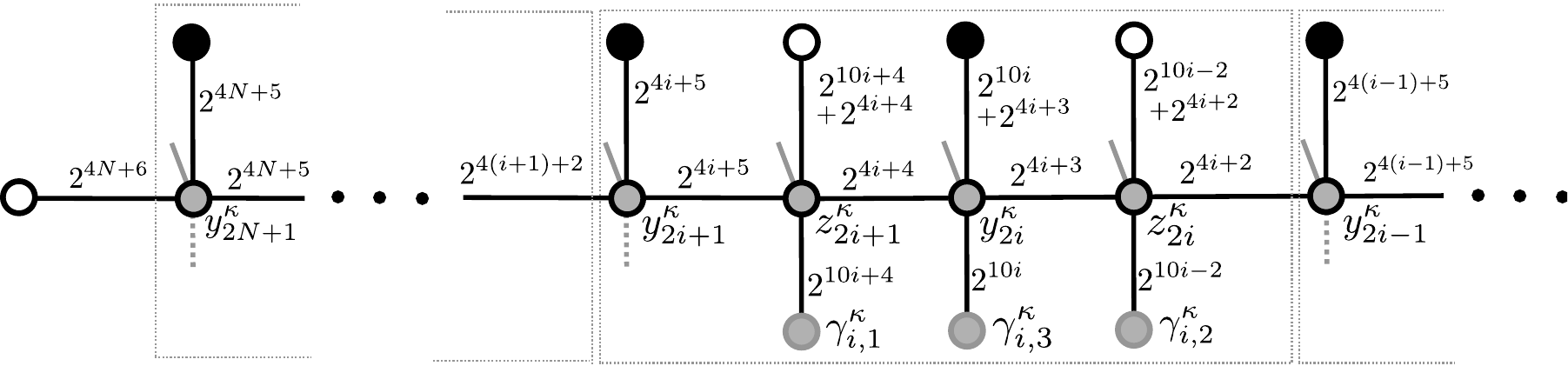}}
				\subfigure{%
				\qquad
					\includegraphics[width=0.59\textwidth]{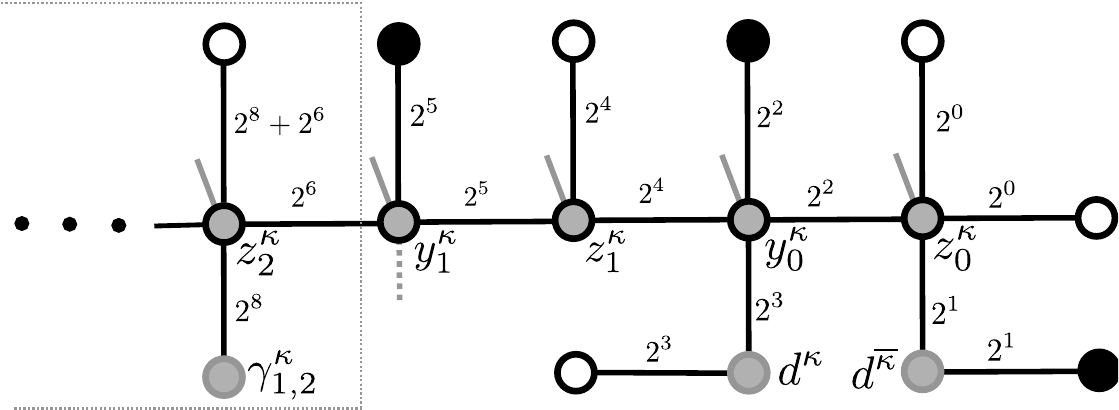}%
					\label{fig: yzende}\qquad}
				\caption{The components of type 4}
				\label{yz}
			\end{figure}		
			
\textbf{Type 5} contains the nodes and edges as depicted in Figure \ref{fig:outputs} for $1 \le i \le m$ and edges $\{1,d^0\}$, $\{1,d^0\}$, $\{0,d^1\}$, and $\{0,d^1\}$ of relative weight $1$ -- these edges are not depicted. The aim of the component is twofold. On the one hand it is to incite that one of the nodes $d^0$ and $d^1$ to become black for which the output of the corresponding copy $G_{C}^{0}$ and $G_{C}^{1}$ is smaller and the other one to become white. On the other hand, the edges $\{1,d^0\}$, $\{1,d^0\}$, $\{0,d^1\}$, and $\{0,d^1\}$ are to break the tie in favor of $G_C^0$ if the outputs of $G_C^0$ and $G_C^1$ are equal.

	\begin{figure}[htb]
		\centering
			\includegraphics[width=0.55\textwidth]{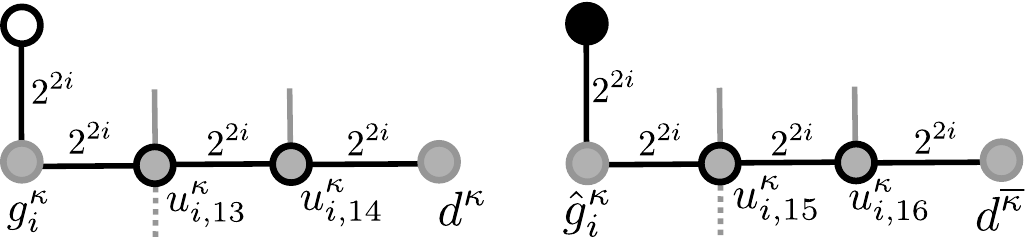}
			\caption{One part of the component of type 5}
			\label{fig:outputs}
	\end{figure}	

\textbf{Type 6} contains nodes $\hat{d}_{i}^{\kappa}$ for all $1 \leq i \leq n$ with incident edges $\{\hat{d}_{i}^{\kappa},d^{\kappa}\}$ of relative weight $2^{2i}$. These edges are to ensure that $col(d^{\kappa}) \neq c(\hat{d}^{\kappa}_{i})$ for all $i$. The component also contains $n$ edges $\{1,d^{\kappa}\}$ with relative weights $2^{2i}$ for all $1 \leq i \leq n$ -- recall that each constant is represented by a separate node of type $1$. These edges are needed for $d^{\kappa}$ to be a comparing node.
			
\textbf{Type 7} (Figure \ref{resetter}) is to incite the input nodes of $G_C^{\overline{\kappa}}$ to take the color corresponding to the better neighbor computed by $G_C^{\kappa}$ if $col(d^{\kappa})=0$. As we will see in Lemma \ref{resetterlemma} the node $\lambda_{i}^{\overline{\kappa}}$ has the same color as $w_{i,1}^{\kappa}$ if $col(w_{i,1}^{\kappa})=col(w_{i,2}^{\kappa})$ and $col(d^{\kappa})=0$. Moreover, we will see in the same Lemma that $\lambda_{i}^{\overline{\kappa}}$ has the opposite color as $\mu_{i}^{\overline{\kappa}}$ in any local optima. Therefore the nodes $\lambda_{i}^{\overline{\kappa}}$ and $\mu_{i}^{\overline{\kappa}}$ together with their incident edges, in the case that $w_{i,1}^{\kappa}=w_{i,2}^{\kappa}$ and $col(d^{\kappa})=0$, have the functionality of a subgraph $T^{\kappa}$ that looks at the nodes $w_{i,1}^{\kappa}$ and biases the input nodes of $G_{C}^{\overline{\kappa}}$ to take the color of their corresponding $w_{i,1}^{\kappa}$. Concerning the maximum degree of five recall that the number of edges of type 3 incident to $x^{\overline{\kappa}}$ was three due to the adjustment. One edge of type 7 is incident to $x^{\overline{\kappa}}$ and one edge of higher type -- depicted as a gray edge in Figure \ref{resetter} -- is incident to $x^{\overline{\kappa}}$. Thus, $x^{\overline{\kappa}}$ has a degree of five.

	\begin{figure}[htb]
		\centering
			\includegraphics[width=0.43\textwidth]{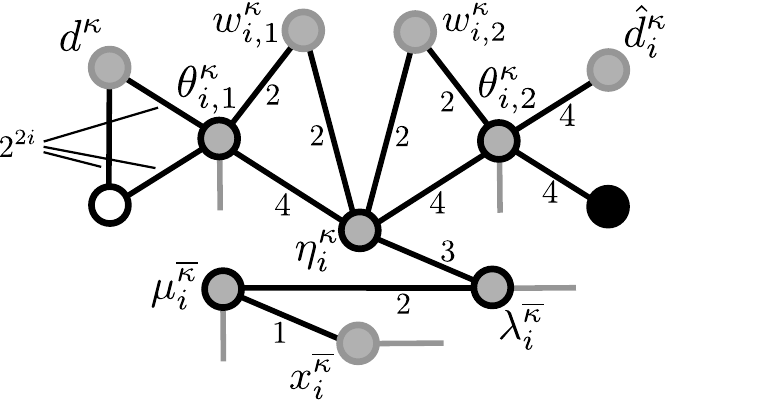}
			\caption{The components of type 7}
			\label{resetter}
	\end{figure}

The components of type 8 to 14 are subgraphs that look at certain nodes and bias other nodes. No node to which any component looks at is a comparing node. Therefore, all of them must be of degree at most five in our construction. But to some of these nodes more than one component looks at. To maintain a maximum degree of five for these nodes, we assume that the component of the lowest type which looks at such a node $v$ not only biases the nodes of which we state that it biases them but also biases extra nodes $v'_1,\ldots,v'_k$, for $k \in \N$ great enough, to have the same color as $v$ and the components of higher types look at $v'_1,\ldots,v'_k$ instead of the original nodes. 
	
\textbf{Type 8} looks at the vectors $x^{0}, x^{1}$ of nodes representing the inputs of $G^{C}_{0}$ and $G^{C}_{1}$ and at the vectors $\lambda^{0}$, $\lambda^{1}$ of nodes of type $7$ and biases the vectors $y^{0}$, $z^{0}$, $y^{1}$, and $z^{1}$ in the following way. The nodes $y_{i}^{0},z_{i}^{0}$ for all $0 \le i \le 2N+1$ are biased to their unnatural value, as defined in type $4$, if $C(x^{0}) < C(x^{1})$, $w(x^{1}) \ne c(x^{0})$, and $w(x^{1}) = c(\lambda^{0})$ and to their natural value otherwise. Similarly, $y_{i}^{1},z_{i}^{1}$ are biased to their unnatural value if $C(x^{0}) \geq C(x^{1})$, $w(x^{0}) \ne c(x^{1})$, and $w(x^{0}) = c(\lambda^{1})$ and to their natural value otherwise. The comparison between $C(x^{0})$ and $C(x^{1})$ is used to decide which circuit is the winner and which one is the loser and the consideration of the other colors is to avoid certain troublemaking local optima.

The idea behind the next two components is as follows. In any local optimum, we want for the nodes $d^{0}$ and $d^{1}$ at most one to be black. The immediate idea to reach this would be to use a simple edge between them in the component of type 2 (see Figure \ref{di}) without the intermediate nodes $u^{0}$ and $u^{1}$. To show -- later in the proof -- that a comparing node $d^{\kappa}$ has a certain color, we want to apply Theorem \ref{substitute}. For this, we need to know the colors of the neighbors adjacent to $d^{\kappa}$ via the edges of the highest weight, which includes the color of $d^{\overline{\kappa}}$. But argue about the color of $d^{\overline{\kappa}}$ via Theorem \ref{substitute} analogously needs the information about the color of $d^{\kappa}$. To solve this problem, we introduce the intermediate nodes $u^{0}$ and $u^{1}$, bias them appropriately and use their colors to bias $d^{0}$ and $d^{1}$.

\textbf{Type 9} looks at $y_1^{0}$, $y_1^{1}$, and at the vectors $x^{0}$ and $x^{1}$ and biases $u^{0}$ and $u^{1}$ as follows. If $C(x^{0}) \geq C(x^{1})$ then it biases $u^{0}$ to the color of $y_1^{0}$ and $u^{1}$ to the opposite. Otherwise it biases $u^{1}$ to the color of $y_1^{1}$ and $u^{0}$ to the opposite.

\textbf{Type 10} looks at $u^{0}$, $u^{1}$, $y_1^{0}$, $y_1^{1}$, and at the vectors $x^{0}$ and $x^{1}$ and biases $d^{0}$ and $d^{1}$ as follows. If $c(y_1^{0})=c(y_1^{1})=0$ then $d^{0}$ is biased to the color of $u^{1}$ and $d^{1}$ to the color of $u^{0}$. If $c(y_1^{0}) \neq c(y_1^{1})$ then $d^{0}$ is biased to the color of $y_1^{1}$ and $d^{1}$ to the opposite. If $c(y_1^{0})=c(y_1^{1})=1$ then we distinguish two cases. If $C(x^{0}) \geq C(x^{1})$ then $d^{0}$ is biased to $0$ and $d^{1}$ to $1$, otherwise $d^{0}$ to $1$ and $d^{1}$ to $0$.

\textbf{Type 11} is to bias the nodes of type 3 to certain preferred colors depending on whether $y_{2i+1}^{\kappa}$ has its natural value. If it has its natural value then it biases the subgraph $S_{i}^{\kappa}$ to colors which reflect the behavior of a NOR-gate for $S_{i}^{\kappa}$ and otherwise it biases them such that the input nodes $I_{1}(g_{i}^{\kappa})$ and $I_{2}(g_{i}^{\kappa})$ are indifferent with respect to their neighbors in $S_{i}^{\kappa}$, i.\ e. the nodes of $S_{i}^{\kappa}$ are biased to their reset state. In particular the component looks at $y_{2i+1}^{\kappa}$ for $1 \leq i \leq N$ and biases $\alpha_{i,1}^{\kappa},$ $\alpha_{i,2}^{\kappa},$ $\gamma_{i,1}^{\kappa},$ $\gamma_{i,2}^{\kappa},$ $\beta_{i,3}^{\kappa},\tau_{i,1}^{\kappa},$ and $\tau_{i,2}^{\kappa}$ to the color of $y_{2i+1}^{\kappa}$ and $\beta_{i,1}^{\kappa},\beta_{i,2}^{\kappa},\gamma_{i,3}^{\kappa},\sigma_{i,1}^{\kappa},\sigma_{i,2}^{\kappa},\delta_{i,1}^{\kappa},$ and $\delta_{i,2}^{\kappa}$ to the opposite. 

The aim of the next two components is as follows. We want to bias the comparing nodes $g_i^{\kappa}$ such that we can apply Theorem \ref{substitute} to obtain either $col(g_i^{\kappa})=1$ or $col(g_i^{\kappa})=0$. To reach this, we need to know the colors of the nodes adjacent to $g_i^{\kappa}$. For this purpose we introduce -- similarly as in the component of type 2 -- extra nodes $u_{i,j}^{\kappa}$, bias them appropriately and use their colors instead. 

\textbf{Type 12} looks at $y_{2i+1}^{\kappa},y_{2i-1}^{\kappa},\alpha_{i,1}^{\kappa},$ and $\alpha_{i,2}^{\kappa}$ and biases $u_{i,1}^{\kappa},$ $u_{i,3}^{\kappa},$ $u_{i,5}^{\kappa},$ $u_{i,7}^{\kappa},$ $u_{i,10}^{\kappa},$ $u_{i,12}^{\kappa}$ to white and $u_{i,2}^{\kappa},u_{i,4}^{\kappa},u_{i,6}^{\kappa},u_{i,8}^{\kappa},u_{i,9}^{\kappa},u_{i,11}^{\kappa}$ to black if $c(y_{2i+1}^{\kappa})=c(y_{2i-1}^{\kappa})$. Otherwise, $u_{i,3}^{\kappa},$ $u_{i,4}^{\kappa},$ $u_{i,7}^{\kappa},$ $u_{i,8}^{\kappa},$ $u_{i,11}^{\kappa},$ $u_{i,12}^{\kappa}$ are biased to their respective opposite and the biases of the remaining nodes split into the following cases. Node $u_{i,1}^{\kappa}$ is biased to $c(\alpha_{i,1}^{\kappa})$ and $u_{i,2}^{\kappa}$ to the opposite. Similarly, $u_{i,5}^{\kappa}$ is biased to $c(\alpha_{i,2}^{\kappa})$ and $u_{i,6}^{\kappa}$ to the opposite. Finally, $u_{i,9}^{\kappa}$ is biased to $c(\alpha_{i,2}^{\kappa}) \wedge c(\alpha_{i,2}^{\kappa})$ and $u_{i,10}^{\kappa}$ to the opposite.

\textbf{Type 13} looks for all $1 \le i \le m$ at $y_{2i-1}^{\kappa},\alpha_{i,1}^{\kappa},$ and $\alpha_{i,2}^{\kappa}$ and biases $u_{i,14}^{\kappa}$ to $c(y_{2i-1}^{\kappa}) \wedge c(\alpha_{i,1}^{\kappa}) \wedge c(\alpha_{i,2}^{\kappa})$ and $u_{i,13}^{\kappa}$ to the opposite. Similarly, it looks for all $m+1 \le i \le 2m$ at $y_{2i-1}^{\kappa},\alpha_{i,1}^{\kappa},$ and $\alpha_{i,2}^{\kappa}$ and biases $u_{i,15}^{\kappa}$ to $c(y_{2i-1}^{\kappa}) \wedge (\neg c(\alpha_{i,1}^{\kappa}) \vee \neg c(\alpha_{i,2}^{\kappa}))$ and $u_{i,16}^{\kappa}$ to the opposite. 

\textbf{Type 14} looks at all nodes of type lower than $14$ that are adjacent to $g_i^{\kappa}$ with the single exception of $\eta_{i}^{\kappa}$ if $g_i^{\kappa}=w_{j,k}^{\kappa}$. Namely, it looks at $u_{i,4j+2}^{\kappa},u_{i,4j+3}^{\kappa}$ for $0 \le j \le 2$, at $\alpha_{j,k}^{\kappa}$ and $\sigma_{j,k}^{\kappa}$ if $I_k(g_j)^{\kappa}=(g_i^{\kappa})$ for $k \in \{1,2\}$, at $u_{i,13}$ if $i \le m$, at $u_{i,15}$ if $m+1 \le i \le 2m$. Furthermore, it looks at $\mu_{i}^{\kappa}$ if $g_i^{\kappa}=w_{j,k}^{\kappa}$, at $\alpha_{i,1}^{\kappa}$, and $\alpha_{i,2}^{\kappa}$. The component treats the color of $\mu_{i}^{\kappa}$ as if it was the color of $\eta_{i}^{\kappa}$ if $g_i^{\kappa}=w_{j,k}^{\kappa}$ -- the component cannot look at $\eta_{i}^{\kappa}$ since $\eta_{i}^{\kappa}$ would have a degree of six in this case but we will see in Lemma \ref{resetterlemma} that $c(\eta_{i}^{\kappa}) = c(\mu_{i}^{\overline{\kappa}})$ in any local optimum. Then, the component computes whether $g_i^{\kappa}$ is weakly indifferent and the color $\rho \in \{0,1\}$ it would have if it was a single node and not weakly indifferent. It biases $g_{i}^{\kappa}$ to $\rho$ if it is not weakly indifferent. If $g_{i}^{\kappa}$ is weakly indifferent then it biases $g_{i}^{\kappa}$ to $c(\alpha_{i,1}^{\kappa}) \wedge c(\alpha_{i,2}^{\kappa})$. This finishes the description of $G_C.$

Now we consider the colors of the nodes of $G_{C}$ in an arbitrary local optimum. All of the remaining Lemmas have an inherent statement ``for any local optimum $P$''. We call a gate $g_{i}^{\kappa}$ \textbf{correct} if $col(g_{i}^{\kappa}) = \neg (col(I_{1}(g_{i}^{\kappa})) \vee col(I_{2}(g_{i}^{\kappa}))$. In the following we will, among other things, argue about the colors of the comparing nodes $v \in V^{C}$ in $P$. We do this by naming the decisive neighbors of $v$, their colors, and the color to which $v$ is biased. Then, we can deduce the color of $v$ via Theorem \ref{substitute} -- recall that a necessary condition of Theorem \ref{substitute} is that $v$ is biased to the opposite color as the color of its decisive neighbors if $v$ is not weakly indifferent. The following Lemmas characterize properties of some components. 

\begin{lemma}\label{compareuuneq}
$d^{0},d^{1}$ and $g_{i}^{\kappa}$ for any $1 \leq i \leq N-n$, $\kappa \in \{0,1\}$ are comparing nodes. Either $col(g_{i}^{\kappa})=1$ or $col(g_{i}^{\kappa})=0$ for all $1 \le i \le N$. Moreover, $c(u^{0}) \neq c(u^{1}).$
\end{lemma}

\begin{proof}
In Table \ref{comptable} we name all nodes adjacent to $d^0,d^1$, and $g_i^{\kappa}$ for all $1 \le i \le N-n$, $\kappa \in \{0,1\}$ and the weigths of the corresponding edges. By means of the table it can easily be verified that the aforementioned nodes are comparing.\newline

\begin{table}[h!b!p!]
\centering
\begin{tabular}{|c||c|c|c|c|}
\hline
Node & Neighbor & Type & R. Weight & Condition\\ \hline
\multirow{15}{*}{$d^{\kappa}$} & $1$ & \multirow{2}{*}{$2$} & \multirow{2}{*}{$1$} & \\
 & $u^{\kappa}$ & & & \\ \cline{2-4}
 & $0$ & \multirow{4}{*}{$4$} & \multirow{2}{*}{$2^{3}$} & \\
 & $y^{\kappa}_0$ & & & \\ \cline{2-2}\cline{4-4}
 & $1$ & & \multirow{2}{*}{$2^{1}$} &\\
 & $z^{\overline{\kappa}}_0$ & & &\\ \cline{2-5}
 & $u^{\kappa}_{i,14}$ & \multirow{4}{*}{$5$} & \multirow{2}{*}{$2^{2i}$} & \multirow{2}{*}{$1 \le i \le m$}\\
 & $u^{\overline{\kappa}}_{i,16}$  & & &\\  \cline{2-2}\cline{4-5}
 & $\kappa$ & & \multirow{2}{*}{$1$} & \\
 & $\kappa$ & & & \\ \cline{2-5}
 & $1$ & \multirow{2}{*}{$6$} & \multirow{2}{*}{$2^{2i}$} & \multirow{2}{*}{$1 \le i \le n$}\\
 & $\hat{d}_{i}^{\kappa}$ & & & \\ \cline{2-5}
 & $\theta^{\kappa}_{i,1}$ & \multirow{2}{*}{$7$} & \multirow{2}{*}{$2^{2i}$} & \multirow{2}{*}{$1 \le i \le m$}\\
 & $0$ & & &\\ 
\cline{2-5}
 & no name & $10$ & $1$ &\\
\hline
\end{tabular}
\hspace{5mm}
\begin{tabular}{|c|c||c|c|c|c|}
\hline
Node & Condition & Neighbor & Type & R. Weight & Condition\\ \hline
\multirow{23}{*}{$g_i^{\kappa}$} &\multirow{6}{*}{$1 \le i \le N-n$} & $u_{i,2}$ & \multirow{14}{*}{$3$} & \multirow{2}{*}{$2^{10i+5}$} &\\
 & & $u_{i,3}$ & & &\\ \cline{3-3}\cline{5-5}
 & & $u_{i,10}$ & & \multirow{2}{*}{$2^{10i+1}$} &\\
 & & $u_{i,11}$ & & &\\ \cline{3-3}\cline{5-5}
 & & $u_{i,6}$ & & \multirow{2}{*}{$2^{10i-1}$} &\\
 & & $u_{i,7}$ & & &\\ \cline{2-3}\cline{5-6}
 & & $1$ & & \multirow{2}{*}{$2^{10j+7}$} & \multirow{4}{*}{$I_1(g_j)=g_i$}\\
 & & $\alpha_{j,1}^{\kappa}$ & & &\\ \cline{3-3}\cline{5-5}
 & & $0$ & & \multirow{2}{*}{$2^{10j+5}$} &\\
 & $2(n+m)+1$ & $\sigma_{j,1}^{\kappa}$ & & &\\ \cline{3-3}\cline{5-6}
 & $\le i \le N-n$ & $0$ & & \multirow{2}{*}{$2^{10j+5}$} & \multirow{4}{*}{$I_2(g_j)=g_i$}\\
 & & $\sigma_{j,2}^{\kappa}$ & & &\\ \cline{3-3}\cline{5-5}
 & & $1$ & & \multirow{2}{*}{$2^{10j-1}$} &\\
 & & $\alpha_{j,2}^{\kappa}$& & &\\ \cline{2-6}
 & \multirow{2}{*}{$1 \le i \le m$} & $0$ & \multirow{4}{*}{$5$} & \multirow{4}{*}{$2^{2i}$} & \\
 & & $u_{i,13}^{\kappa}$& & &\\ \cline{2-3}
 & \multirow{2}{*}{$m+1 \le i \le 2m$} & $1$ & & & \\
 & & $u_{i,15}^{\kappa}$& & &\\ \cline{2-6}
 & $2m+n+1 \le$ & $\theta_{j,2}^{\kappa}$ & \multirow{4}{*}{$7$} & \multirow{4}{*}{$2$} & \multirow{4}{*}{$j:=i-2m-n$}\\
 & $i \le 2(m+n)$ & $\eta_{j}^{\kappa}$ & & &\\ \cline{2-3}
 & $2m+1 \le$ & $\theta_{j,1}^{\kappa}$ & & &\\
 & $i \le 2m+n$ & $\eta_{j}^{\kappa}$ & & &\\ \cline{2-6}
 & & no name & $14$ & $1$ &\\
\hline
\end{tabular}
\caption{Neighborhood of the nodes $d^0, d^1$, and $g_i^{\kappa}$ for $1 \le i \le N-n$, $\kappa \in \{0,1\}.$}
\label{comptable}
\end{table}

Now consider the nodes $g_{i}^{\kappa}$. Recall first that Theorem \ref{substitute} only applies to local optima in which the comparing node is biased to the color that it had if it was a single node. The only nodes different from the constants that are incident to any $g_{j}^{\kappa}$ and to which the components of type $14$ do not look at is $\mu_{i}^{\overline{\kappa}}$ for $g_{j}^{\kappa}=w_{i,k}^{\kappa}$ and any $k \in \{1,2\}$. From Lemma \ref{resetterlemma} we know that $c(\eta_{i}^{\kappa}) = c(\mu_{i}^{\overline{\kappa}})$. Thus, the component of type $14$ correctly decides whether $g_i^{\kappa}$ is weakly indifferent as outlined in the description of type $14$ and therefore it biases $g_i^{\kappa}$ such that Theorem \ref{substitute} implies that either $col(g_{i}^{\kappa})=1$ or $col(g_{i}^{\kappa})=0$ for all $1 \le i \le N$.

Due to the weights of the edges incident to $u^0$ and $u^1$ and since they are biased to different colors by type $9$ in each local optimum at least one of them is unhappy if both have the same color. Thus, the claim follows.
\end{proof}

\begin{lemma}[similar to Claims 5.9.B and 5.10.B in \cite{schaeffer1}]\label{resetterlemma}
If $col(d^{\kappa}) = 1$ then neither flipping $w_{i,1}^{\kappa}$ nor $w_{i,2}^{\kappa}$ change the cut by a weight of type $7$. If $col(d^{\kappa})=0$ and $col(w_{i,1}^{\kappa})=col(w_{i,2}^{\kappa})$ then $col(w_{i,1}^{\kappa}) \neq c(\eta_{i}^{\kappa})$. Moreover, $c(\eta_{i}^{\kappa}) \ne c(\lambda_{i}^{\overline{\kappa}}) \ne c(\mu_{i}^{\overline{\kappa}}).$
\end{lemma}

\begin{proof}

The proof uses the following claim.

\begin{claim}\label{uneqlemma}
If $col(d^{\kappa}) = \rho$ for $\rho \in \{0,1\}$ then $c(\hat{d}^{\kappa}_i)=\overline{\rho}$ for all $1 \leq i \leq n$.
\end{claim}
\begin{proof}
There are three edges incident to each node $\hat{d}^{\kappa}_i$ as introduced in type $6$. Namely, one edge of type $6$ and two edges of type $7$. Since the weight of the edge of type $6$ is greater than the sum of all edges of higher type, in particular the two edges of type $7$, the claim follows.
\end{proof}

Assume $col(d^{\kappa})=1$. Then, by Claim \ref{uneqlemma} we have $c(\hat{d}^{\kappa}_{i})=0$ for all $i$. Since $col(d^{\kappa})=1$, the weights of the five edges incident to $\theta_{i,1}^{\kappa}$ as depicted in Figure \ref{resetter} imply $c(\theta_{i,1}^{\kappa}) \neq c(\eta_{i}^{\kappa})$. Similarly, we can argue that $c(\theta_{i,2}^{\kappa}) \neq c(\eta_{i}^{\kappa})$. But then, neither a flip of $w_{i,1}^{\kappa}$ nor a flip of $w_{i,2}^{\kappa}$ can change the cut by a weight of type $7$.

Now assume $col(d^{\kappa})=0$ and $col(w_{i,1}^{\kappa})=col(w_{i,2}^{\kappa})$. Due to Claim \ref{uneqlemma} we have $c(\hat{d}^{\kappa}_{i})=1$ for all $i$. The weights of the edges incident to $\theta_{i,1}^{\kappa}$ and $\theta_{i,2}^{\kappa}$ imply $c(\theta_{i,1}^{\kappa})=1$ and $c(\theta_{i,2}^{\kappa})=0$. Since $col(w_{i,1}^{\kappa})=col(w_{i,2}^{\kappa})$ and $c(\theta_{i,1}^{\kappa}) \neq c(\theta_{i,2}^{\kappa})$, node $\eta_i^{\kappa}$ is happy if and only if its color is different from the color of $w_{i,1}^{\kappa}$ and $w_{i,2}^{\kappa}$.

Finally, the claim $c(\eta_{i}^{\kappa}) \ne c(\lambda_{i}^{\overline{\kappa}}) \ne c(\mu_{i}^{\overline{\kappa}})$ follows directly from the weights of the edges incident to $\lambda_{i}^{\overline{\kappa}}$ and $\mu_{i}^{\overline{\kappa}}$.
\end{proof}

\begin{lemma}[similar to Lemma 4.1H in \cite{schaeffer1}]\label{propagation}
If $c(z_{j}^{\kappa})=1$ then $c(y_{j-1}^{\kappa})=0$. If $c(y_{j}^{\kappa})=0$ then $c(y_{p}^{\kappa})=0$ and $c(z_{p}^{\kappa})=1$ for all $p \le j$.
\end{lemma}

\begin{proof}
The sum of the weights of the edges $\{z_{j}^{\kappa},y_{j-1}^{\kappa}\}$ and $\{y_{j-1}^{\kappa},1\}$ is greater than the sum of all other edges incident to $y_{j-1}^{\kappa}$. Thus, if $c(z_{j}^{\kappa})=1$ then $c(y_{j-1}^{\kappa})=0.$ Similarly, we can argue that $z_{p}^{\kappa}$ has its unnatural value if $y_{p}^{\kappa}$ has its unnatural value. Therefore, the claim follows by induction.
\end{proof}

\begin{lemma}[similar to Lemma 4.1 in \cite{schaeffer1}]\label{incorrectnesslemma}
If $g_{i}^{\kappa}$ is not correct then $c(z_{2i}^{\kappa})=1$.  
\end{lemma}

\begin{proof} The proof uses the following claims.

\begin{claim}\label{ZNatErgoYNat}
If $c(z_{2i}^{\kappa})=0$ then $c(y_{2i-1}^{\kappa})=1$.
\end{claim}
\begin{proof}
Assume $c(z_{2i}^{\kappa})=c(y_{2i-1}^{\kappa})=0$. If $y_{2i-1}^{\kappa}$ is biased to black by the component of type $8$ then $c(y_{2i-1}^{\kappa})=1$ since $c(z_{2i}^{\kappa})=0$ which is a contradiction. Thus, $y_{2i-1}^{\kappa}$ is biased to $0$. Since $z_{2i}^{\kappa}$ and $y_{2i-1}^{\kappa}$ are biased to opposite colors by type $8$, node $z_{2i}^{\kappa}$ is biased to $1$. Due to the weight of its incident edges it cannot be white then. But this is a contradiction.
\end{proof}

\begin{claim}\label{nr1} If $col(I_1(g_i^{\kappa}))=1$ then $c(\alpha_{i,1}^{\kappa})=0$ and $c(\beta_{i,1}^{\kappa})=1$. If $col(I_2(g_i^{\kappa}))=1$ then $c(\alpha_{i,2}^{\kappa})=0$ and $c(\beta_{i,2}^{\kappa})=1$.
\end{claim}
\begin{proof}
If $col(I_1(g_i^{\kappa}))=1$ then $c(\alpha_{i,1}^{\kappa})=0$ since the edges $\{I_1(g_i^{\kappa}),\alpha_{i,1}^{\kappa}\}$ and $\{\alpha_{i,1}^{\kappa},1\}$ combined weigh more than the sum of all other edges incident to $\alpha_{i,1}^{\kappa}$. Analogously, $c(\alpha_{i,1}^{\kappa})=0$ implies $c(\beta_{i,1}^{\kappa})=1$. The argumentation for the second part of the claim is similar.
\end{proof}

\begin{claim}\label{nr2} If $col(I_1(g_i^{\kappa}))=1$ and $col(g_{i}^{\kappa})=1$ then $c(z_{2i}^{\kappa})=1$. If $col(I_2(g_i^{\kappa}))=1$ and $col(g_{i}^{\kappa})=1$ then $c(z_{2i}^{\kappa})=1$.
\end{claim}
\begin{proof}
Assume for the sake of contradiction $col(I_1(g_i^{\kappa}))=1$, $col(g_{i}^{\kappa})=1$, but $c(z_{2i}^{\kappa})=0$.  Claim \ref{ZNatErgoYNat} implies $c(y_{2i-1}^{\kappa})=1$ since $c(z_{2i}^{\kappa})=0$. Moreover, Lemma \ref{propagation} implies $c(y_{2i+1}^{\kappa})=1$ since $c(z_{2i}^{\kappa})=0$. Thus, $c(y_{2i+1}^{\kappa})=c(y_{2i-1}^{\kappa})$ and therefore the nodes $u_{i,3}^{\kappa}$ and $u_{i,4}^{\kappa}$ are biased to $0$ and $1$ respectively by the component of type $12$. Then, $c(u_{i,3}^{\kappa})=0$ and therefore $c(u_{i,4}^{\kappa})=1$. Then, Claim \ref{nr1} implies $c(\gamma_{i,1}^{\kappa})=0$ and thereafter $c(z_{2i+1}^{\kappa})=1$. Then, $c(z_{2i}^{\kappa})=1$ due to Lemma \ref{propagation} which is a contradiction. The proof for $col(I_2(g_i^{\kappa}))=1$ is analogous.
\end{proof}

\begin{claim}\label{nr4} If $col(I_1(g_i^{\kappa}))=0$ then $c(\delta_{i,1}^{\kappa})=1$. If $col(I_2(g_i^{\kappa}))=0$ then $c(\delta_{i,2}^{\kappa})=1$. 
\end{claim}
\begin{proof}
If $col(I_1(g_i^{\kappa}))=0$ then $c(\sigma_{i,1}^{\kappa})=1$ since the edges $\{I_1(g_i^{\kappa}),\sigma_{i,1}^{\kappa}\}$ and $\{\sigma_{i,1}^{\kappa},0\}$ combined weigh more than the sum of all other edges incident to $\sigma_{i,1}^{\kappa}$. Similarly, it follows that $c(\tau_{i,1}^{\kappa})=0$ and $c(\delta_{i,1}^{\kappa})=1$. Analogously, $col(I_2(g_i^{\kappa}))=0$ implies $c(\delta_{i,2}^{\kappa})=1$.
\end{proof}

\begin{claim}\label{nr5} If $col(I_1(g_i^{\kappa}))=col(I_2(g_i^{\kappa}))=0$ then $c(\beta_{i,3}^{\kappa})=0$. 
\end{claim}
\begin{proof}
Due to Claim  \ref{nr4}, $c(\delta_{i,1}^{\kappa})=c(\delta_{i,2}^{\kappa})=1$. Since the sum of the weights of the edges $\{\beta_{i,3}^{\kappa},\delta_{i,1}^{\kappa}\}$ and $\{\beta_{i,3}^{\kappa},\delta_{i,2}^{\kappa}\}$ is greater than the sum of all other edges incident to $\beta_{i,3}^{\kappa}$ the claim follows.
\end{proof}

\begin{claim}\label{nr6} If $col(I_1(g_i^{\kappa}))=col(I_2(g_i^{\kappa}))=col(g_i^{\kappa})=0$ then $c(z_{2i}^{\kappa})=1$. 
\end{claim}
\begin{proof}
Assume $col(I_1(g_i^{\kappa}))=col(I_2(g_i^{\kappa}))=col(g_i^{\kappa})=0$ but $c(z_{2i}^{\kappa})=0$. Then, Lemma \ref{propagation} implies $c(y_{2i+1}^{\kappa})=1$ since $c(z_{2i}^{\kappa})=0$. Moreover, Claim \ref{ZNatErgoYNat} implies $c(y_{2i-1}^{\kappa})=1$. Thus, $c(y_{2i+1}^{\kappa})=c(y_{2i-1}^{\kappa})$ and therefore the nodes $u_{i,11}^{\kappa}$ and $u_{i,12}^{\kappa}$ are biased to $1$ and $0$ respectively by the component of type $12$. Then, $c(u_{i,11}^{\kappa})=1$ and therefore $c(u_{i,12}^{\kappa})=0$. But then Claim \ref{nr5} implies $c(\gamma_{i,3}^{\kappa})=1$ and therefore $c(y_{2i}^{\kappa})=0$. Then, $c(z_{2i}^{\kappa})=1$ due to Lemma \ref{propagation} which is a contradiction. Thus, the claim follows.
\end{proof}

Lemma \ref{compareuuneq} implies that either $col(g_{i}^{\kappa})=1$ or $col(g_{i}^{\kappa})=0$. If $col(I_1(g_i^{\kappa}))=1$ or $col(I_2(g_i^{\kappa}))=1$ then the claim follows from Claim \ref{nr2}. If $col(I_1(g_i^{\kappa}))=col(I_2(g_i^{\kappa}))=0$ then the claim follows from Claim \ref{nr6}.
\end{proof}

\begin{lemma}[partially similar to  Lemma 4.2 in \cite{schaeffer1}]\label{listeninglemma}
If $c(y_{2i+1}^{\kappa})=0$ then $c(\alpha_{i,1}^{\kappa})=c(\alpha_{i,2}^{\kappa})=0$ and $c(\sigma_{i,1}^{\kappa})=c(\sigma_{i,2}^{\kappa})=1$.
\end{lemma}

\begin{proof} Assume $c(y_{2i+1}^{\kappa})=0$. From Lemma \ref{propagation} we know that $c(z_{2i+1}^{\kappa})=c(z_{2i}^{\kappa})=1$ and $c(y_{2i}^{\kappa})=c(y_{2i-1}^{\kappa})=0$. We proof is done by means of the following two claims.

\begin{claim}\label{contradiction} Assume $c(y_{2i+1}^{\kappa})=0$. Then, $c(\alpha_{i,1}^{\kappa})=c(\alpha_{i,2}^{\kappa})=0,c(\beta_{i,1}^{\kappa})=c(\beta_{i,2}^{\kappa})=1,$ and $c(\gamma_{i,1}^{\kappa})=c(\gamma_{i,2}^{\kappa})=0$. 
\end{claim}
\begin{proof}
From Lemma \ref{compareuuneq} we know that either $col(g_i^{\kappa})=1$ or $col(g_i^{\kappa})=0$. From the component of type 11 node $\beta_{i,1}^{\kappa}$ is biased to $1$ and $\gamma_{i,1}^{\kappa}$ is biased to $0$. 

Assume first $c(\beta_{i,1}^{\kappa}) = c(\gamma_{i,1}^{\kappa}) = 1$. Then, node $\gamma_{i,1}^{\kappa}$ is unhappy since $c(z_{2i+1})=1$ due to Lemma \ref{propagation}. Now assume $c(\beta_{i,1}^{\kappa}) = c(\gamma_{i,1}^{\kappa}) = 0.$ Then, node $\beta_{i,1}^{\kappa}$ is unhappy. Now assume  $c(\beta_{i,1}^{\kappa})=0$ and $c(\gamma_{i,1}^{\kappa})=1$. If $col(g_i^{\kappa})=0$ then $c(u_{i,2}^{\kappa})=1$ and $c(u_{i,1}^{\kappa})=0$ due to their bias from the component of type $12$ -- recall that $c(y_{2i+1}^{\kappa})=c(y_{2i-1}^{\kappa})$ due to Lemma \ref{propagation}. But then $\beta_{i,1}^{\kappa}$ is unhappy, which is a contradiction. Now assume $col(g_i^{\kappa})=1$. Then, $c(u_{i,3}^{\kappa})=0$ and $c(u_{i,4}^{\kappa})=1$ due to the bias of type $12$. But then $\gamma_{i,1}^{\kappa}$ is unhappy since $c(z_{2i+1}^{\kappa})=1$ due to Lemma \ref{propagation} which is also a contradiction. Thus, $c(\beta_{i,1}^{\kappa})=1$ and $c(\gamma_{i,1}^{\kappa})=0$. 

Since $c(\beta_{i,1}^{\kappa})=1$ node $\alpha_{i,1}^{\kappa}$ must be white since it is biased to white by type $11$. The proof for $\alpha_{i,1}^{\kappa},\beta_{i,2}^{\kappa}$, and $\gamma_{i,2}^{\kappa}$ is analogous.
\end{proof}

\begin{claim} Assume $c(y_{2i+1}^{\kappa})=0$. Then, $c(\sigma_{i,1}^{\kappa})=c(\sigma_{i,2}^{\kappa})=1,c(\beta_{i,3}^{\kappa})=0,$ and $c(\gamma_{i,3}^{\kappa})=1$. 
\end{claim}
\begin{proof}
Assume first that $c(\delta_{i,1}^{\kappa})=c(\delta_{i,2}^{\kappa})=0$. Type $11$ biases $\delta_{i,1}^{\kappa}$ and $\delta_{i,2}^{\kappa}$ to black. Therefore, both nodes $\delta_{i,1}^{\kappa}$ and $\delta_{i,2}^{\kappa}$ are unhappy. Therefore, we may assume that at least one of them is black. 

If $c(\beta_{i,3}^{\kappa}) = c(\gamma_{i,3}^{\kappa}) = 1$ then $\beta_{i,3}^{\kappa}$ is unhappy because $\beta_{i,3}^{\kappa}$ is biased to $0$ by type $11$. Now assume $c(\beta_{i,3}^{\kappa}) = c(\gamma_{i,3}^{\kappa}) = 0.$ Then, node $\gamma_{i,3}^{\kappa}$ is unhappy since $c(y_{2i})=0$ has its unnatural value due to Lemma \ref{propagation} and since $\gamma_{i,3}^{\kappa}$ is biased to $1$ by type $11$. Now assume $c(\beta_{i,3}^{\kappa})=1$ and $c(\gamma_{i,3}^{\kappa})=0$. If $col(g_i^{\kappa})=0$ then the bias of type $12$ implies $c(u_{i,11}^{\kappa})=1$ and $c(u_{i,12}^{\kappa})=0$ which is a contradiction since $\gamma_{i,3}^{\kappa}$ is unhappy then due to the bias of type $11$. But if $col(g_i^{\kappa})=1$ then the bias of type $12$ implies $c(u_{i,10}^{\kappa})=0$ and $c(u_{i,9}^{\kappa})=1$ which is also a contradiction since $\beta_{i,3}^{\kappa}$ is unhappy then due to the bias of type $11$. Thus, $c(\beta_{i,3}^{\kappa})=0$ and $c(\gamma_{i,3}^{\kappa})=1$. 

Since $c(\beta_{i,3}^{\kappa})=0$ we get $c(\delta_{i,1}^{\kappa})=c(\delta_{i,2}^{\kappa})=1$ due to the biases of type $11$. Then, $c(\tau_{i,1}^{\kappa})=c(\tau_{i,2}^{\kappa})=0$ and therefore $c(\sigma_{i,1}^{\kappa})=c(\sigma_{i,2}^{\kappa})=1$ also due to the biases of type $11$.
\end{proof}
\end{proof}

\begin{lemma}[partially similar to Lemma 4.3 in \cite{schaeffer1}]\label{naturallemma}
Assume $c(y_{2i+1}^{\kappa})=1$ and $c(y_{2i-1}^{\kappa})=0$. If $g_{i}^{\kappa}$ is correct then $z_{2i}^{\kappa},$ $z_{2i+1}^{\kappa}$, and $y_{2i}^{\kappa}$ have the colors to which they are biased by type $8$. If $g_{i}^{\kappa}$ is not correct then flipping $g_{i}^{\kappa}$ does not decrease the cut by a weight of an edge type $3$ corresponding to $g_{i}^{\kappa}$ and increases it by a weight of type $14$ if $g_{i}$ is indifferent with respect to edges of type $5$ and $7$.
\end{lemma}

\begin{proof} 
The proof uses the following three claims.
\begin{claim}\label{beta12correct}
Assume $c(y_{2i+1}^{\kappa})=1$. Then, $c(\alpha_{i,1}^{\kappa})= \neg col(I_1(g_i^{\kappa}))$ and $c(\alpha_{i,2}^{\kappa})= \neg col(I_2(g_i^{\kappa}))$. If, in addition, $c(y_{2i-1}^{\kappa})=0$ then $c(\beta_{i,1}^{\kappa})= col(I_1(g_i^{\kappa}))$ and $c(\beta_{i,2}^{\kappa})= col(I_2(g_i^{\kappa}))$.
\end{claim}
\begin{proof}
If $col(I_1(g_i^{\kappa}))=1$ then $c(\alpha_{i,1}^{\kappa})=0$. If, on the other hand, $col(I_1(g_i^{\kappa}))=0$ then $c(\alpha_{i,1}^{\kappa})=1$ since $\alpha_{i,1}^{\kappa}$ is biased to $1$ by type $11$. Now assume $c(y_{2i-1}^{\kappa})=0$. If $col(I_1(g_i^{\kappa}))=1$ then $c(\beta_{i,1}^{\kappa})=1$ since $c(\alpha_{i,1}^{\kappa})=0$. Now assume $col(I_1(g_i^{\kappa}))=0$. Due to $c(\alpha_{i,1}^{\kappa})=1$ and since $\beta_{i,1}^{\kappa}$ is biased to $0$ by type $11$, it can only be black if $\gamma_{i,1}^{\kappa}$ and $u_{i,1}^{\kappa}$ are both white. But if $\gamma_{i,1}^{\kappa}$ is white then $u_{i,4}^{\kappa}$ must be black since $\gamma_{i,1}^{\kappa}$ is biased to black by type $11$. If $col(g_i^{\kappa})=1$ then $c(u_{i,2}^{\kappa})=0$ and $c(u_{i,1}^{\kappa})=1$ due to the bias of type $12$ which is a contradiction. On the other hand, if $col(g_i^{\kappa})=0$ then $c(u_{i,3}^{\kappa})=1$ and $c(u_{i,4}^{\kappa})=0$ due to the bias of type $12$ which is also a contradiction. Thus, $c(\beta_{i,1}^{\kappa})=0$.

The argumentation for $\alpha_{i,2}^{\kappa}$ and $\beta_{i,2}^{\kappa}$ is analogous.
\end{proof}

\begin{claim}\label{beta3correct}
Assume $c(y_{2i+1}^{\kappa})=1$ and $c(y_{2i-1}^{\kappa})=0$. Then, $c(\beta_{i,3}^{\kappa})= col(I_1(g_i^{\kappa})) \vee col(I_2(g_i^{\kappa}))$.
\end{claim}
\begin{proof}
If an input is white then the corresponding $\delta_{i,j}^{\kappa}$ is black due to Claim \ref{nr4}. Thus, if both inputs are white then $\beta_{i,3}^{\kappa}$ is white. 

Now assume that at least one input is black. Let $I_{1}(g_{i}^{\kappa})=1$. Since $\sigma_{i,1}^{\kappa}$ is biased to white, we have $c(\sigma_{i,1}^{\kappa}) = 0.$ Analogously, we get $c(\tau_{i,1}^{\kappa}) = 1$. Node $\delta_{i,1}^{\kappa}$ is biased to white by type $11$. If both nodes $\delta_{i,1}^{\kappa}$ and $\delta_{i,2}^{\kappa}$ are black then $\delta_{i,1}^{\kappa}$ is unhappy. Thus, we may assume that at least one of them is white. Since $\beta_{i,3}^{\kappa}$ is biased to 1 by type $11$, it can only be white if $\gamma_{i,3}^{\kappa}$ and $u_{i,9}^{\kappa}$ are both black. But if $\gamma_{i,3}^{\kappa}$ is black then $u_{i,12}^{\kappa}$ must be white since $\gamma_{i,3}^{\kappa}$ is biased to white by type $11$. Then, the bias of type $12$ implies that if $g_i^{\kappa}$ is white then $u_{i,10}^{\kappa}$ is black and $u_{i,9}^{\kappa}$ is white and if $g_i^{\kappa}$ is black then $u_{i,11}^{\kappa}$ is white and $u_{i,12}^{\kappa}$ is black, each resulting in a contradiction. Thus, $c(\beta_{i,3}^{\kappa})=1$.
\end{proof}

\begin{claim}\label{rest}
Assume $c(y_{2i+1}^{\kappa})=1$ and $c(y_{2i-1}^{\kappa})=0$. If $g_i^{\kappa}$ is correct then $c(\gamma_{i,1}^{\kappa})=c(\gamma_{i,2}^{\kappa})=1$ and $c(\gamma_{i,3}^{\kappa})=0$. If $g_i^{\kappa}$ is not correct then at least one of the nodes $c(u_{i,2}^{\kappa}),$ $c(u_{i,3}^{\kappa})$ has the same color as $g_i^{\kappa}$, at least one of the nodes $c(u_{i,6}^{\kappa}),$ $c(u_{i,7}^{\kappa})$ has the same color as $g_i^{\kappa}$, and at least one of the nodes $c(u_{i,10}^{\kappa}),$ $c(u_{i,11}^{\kappa})$ has the same color as $g_i^{\kappa}$.
\end{claim}
\begin{proof}
Assume first that $g_i^{\kappa}$ is correct. From Claim \ref{beta12correct} we know that $c(\beta_{i,1}^{\kappa})=col(I_1(g_i^{\kappa}))$. Since $g_i^{\kappa}$ is correct, at least one of the two nodes $\beta_{i,1}^{\kappa}$ and $g_i^{\kappa}$ is white. Assume first that $c(\beta_{i,1}^{\kappa})=1$. Then, due to Claim \ref{beta12correct}, we have $c(\alpha_{i,1}^{\kappa})=0$. If $g_i^{\kappa}$ is white then $c(u_{i,3}^{\kappa})=1$ and $c(u_{i,4}^{\kappa})=0$ since they are biased to $1$ and $0$ respectively by type $12$. Since at least one of the nodes $u_{i,4}^{\kappa}$ and $\beta_{i,1}^{\kappa}$ is white and $\gamma_{i,1}^{\kappa}$ is biased to black by type $11$ it is actually black. Analogously, we can argue that $\gamma_{i,2}^{\kappa}$ is also black. Moreover, by Claim \ref{beta3correct} we know that $c(\beta_{i,3}^{\kappa})= col(I_1(g_i^{\kappa})) \vee col(I_2(g_i^{\kappa}))$. Since $g_i^{\kappa}$ is correct, it has the opposite color as $\beta_{i,3}^{\kappa}$. If $col(g_i^{\kappa})=1$ then $c(\alpha_{i,1}^{\kappa})=c(\alpha_{i,2}^{\kappa})=1$ and therefore $c(u_{i,11}^{\kappa})=0$ and $c(u_{i,12}^{\kappa})=1$ since they are biased to $0$ and $1$ respectively by type $12$. Therefore, at least one of the nodes $u_{i,12}^{\kappa}$ and $\beta_{i,3}^{\kappa}$ is black. Thus, $\gamma_{i,3}^{\kappa}$ has the color to which it is biased by type $11$, i.\ e. $0$.

Now assume that $g_i^{\kappa}$ is not correct. If $col(I_1(g_i^{\kappa}))=1$ then $c(\alpha_{i,1}^{\kappa})=0$ and $c(\beta_{i,1}^{\kappa})=1$ due to Claim \ref{beta12correct}. Moreover, since $g_i^{\kappa}$ is not correct, we have $col(g_i^{\kappa})=1$. Then $c(\alpha_{i,1}^{\kappa})=0$ and the biases of type 12 imply $c(u_{i,1}^{\kappa})=0$ and $c(u_{i,2}^{\kappa})=1$. If $col(I_1(g_i^{\kappa}))=0$ then $c(\alpha_{i,1}^{\kappa})=1$ and $c(\beta_{i,1}^{\kappa})=0$ due to Claim \ref{beta12correct}. Since $\gamma_{i,1}^{\kappa}$ is biased to $1$ by type $11$ we get $c(\gamma_{i,1}^{\kappa})=1$. Moreover, since $c(\alpha_{i,1}^{\kappa})=1$ the biases of type $12$ imply $c(u_{i,1}^{\kappa})=1,$ $c(u_{i,2}^{\kappa})=0,$ $c(u_{i,4}^{\kappa})=0$ and $c(u_{i,3}^{\kappa})=1$. The proof for $c(u_{i,6}^{\kappa})$ and $c(u_{i,7}^{\kappa})$ is analogous. By Claim \ref{beta3correct} we know that $c(\beta_{i,3}^{\kappa})= col(I_1(g_i^{\kappa})) \vee col(I_2(g_i^{\kappa}))$. Since $g_i^{\kappa}$ is not correct, we have $col(g_i^{\kappa})=c(\beta_{i,3}^{\kappa})$. If $c(\beta_{i,3}^{\kappa})=0$ then, due to Claim \ref{beta12correct} we have $c(\alpha_{i,1}^{\kappa})=c(\alpha_{i,2}^{\kappa})=1$. Then, the biases of the component of type $12$ imply $c(u_{i,9}^{\kappa})=1$ and $c(u_{i,10}^{\kappa})=0$. Thus, $u_{i,10}^{\kappa}$ has the same color as $g_i^{\kappa}$. If $c(\beta_{i,3}^{\kappa})=1$ then $c(\gamma_{i,3}^{\kappa})=0$ since it is biased to white by type $11$. Moreover, $c(\alpha_{i,1}^{\kappa})=0$ or $c(\alpha_{i,2}^{\kappa})=0$ due to Claim \ref{beta12correct}. Then, the biases of the component of type $12$ imply $c(u_{i,9}^{\kappa})=0$ and $c(u_{i,10}^{\kappa})=1$ as well as $c(u_{i,12}^{\kappa})=1$ and $c(u_{i,11}^{\kappa})=0$. Then, we have $c(u_{i,10}^{\kappa}) \neq c(u_{i,11}^{\kappa})$ which proves the claim.
\end{proof}

Assume $c(y_{2i+1}^{\kappa})=1$ and $c(y_{2i-1}^{\kappa})=0$. Assume furthermore that $g_i^{\kappa}$ is correct. Then, due to Claim \ref{rest} we have $c(\gamma_{i,1}^{\kappa})=c(\gamma_{i,2}^{\kappa})=1$ and $c(\gamma_{i,3}^{\kappa})=0$. Then, if the nodes $y_{j}^{\kappa},z_{j}^{\kappa}$ for all $j$ are biased to their natural values then due to $c(y_{2i+1}^{\kappa})=1$ we get $c(z_{2i+1}^{\kappa})=0$, $c(y_{2i}^{\kappa})=1$, and $c(z_{2i}^{\kappa})=0$. If, on the other hand, the nodes $y_{j}^{\kappa},z_{j}^{\kappa}$ for all $j$ are biased to their unnatural values then due to $c(y_{2i-1}^{\kappa})=0$ we get $c(z_{2i}^{\kappa})=1$, $c(y_{2i}^{\kappa})=0$, and $c(z_{2i+1}^{\kappa})=1$. 

Now assume that $g_i^{\kappa}$ is not correct. Due to $c(y_{2i-1}^{\kappa})=0$ Lemma \ref{propagation} implies $c(y_{2j+1}^{\kappa})=0$ for all $j < i$. Then, Lemma \ref{listeninglemma} implies $c(\alpha_{j,1}^{\kappa})=c(\alpha_{j,2}^{\kappa})=0$ and $c(\sigma_{j,1}^{\kappa})=c(\sigma_{j,2}^{\kappa})=1$ for all $j < i.$ Then, Claim \ref{rest} implies that flipping $g_{i}^{\kappa}$ does not decrease the cut by a weight of type $3$. Finally, Claim \ref{beta12correct} implies $c(\alpha_{i,j}^{\kappa})= \neg col(I_j(g_i^{\kappa}))$ for $1 \le j \le 2$. Thus, flipping $g_{i}^{\kappa}$ to its correct color gains a weight of type $14$ if $g_{i}^{\kappa}$ is indifferent with respect to edges of type $5$ and $7$.
\end{proof}

\begin{lemma}\label{onezero}
If $col(d^{\kappa}) = 1$, $col(d^{\overline{\kappa}}) = 0$, and all nodes $y_i^{\kappa},z_i^{\kappa}$ for $0 \le i \le 2N+1$ are biased to their natural values then $c(y_1^{\kappa})=1$.
\end{lemma}

\begin{proof}
Assume $col(d^{\kappa}) = 1$, $col(d^{\overline{\kappa}}) = 0$, and that all nodes $y_i^{\kappa},z_i^{\kappa}$ for $0 \le i \le 2N+1$ are biased to their natural values. We show that all gates of $G_C^{\kappa}$  are correct. For the sake of contradiction we assume that $G_C^{\kappa}$ contains an incorrect gate and let $g_i^{\kappa}$ be the incorrect gate with the highest index. 

We first show by induction that the nodes $y_j^{\kappa},z_j^{\kappa}$ for $j > 2i+1$ and $y_{2i+1}^{\kappa}$ have their natural values. Since $y_{2N+1}^{\kappa}$ is biased to its natural value, we have $c(y_{2N+1}^{\kappa})=1$. Assume $c(y_{2j+1}^{\kappa})=1$ for any $j > i$. If any one of the nodes $z_{2j+1}^{\kappa},y_{2j}^{\kappa},z_{2j}^{\kappa}$ has its unnatural value then Lemma \ref{propagation} implies $c(y_{2j-1}^{\kappa}) = 0$. Then, Lemma \ref{naturallemma} implies that all nodes $z_{2j+1}^{\kappa},y_{2j}^{\kappa},z_{2j}^{\kappa}$ have their natural values whereafter Claim \ref{ZNatErgoYNat} implies $c(y_{2i-1}^{\kappa})=1$ which is a contradiction. Thus, $c(y_{2j+1}^{\kappa})=1$ implies $c(y_{2j-1}^{\kappa})=1$ for any $j > i$ and therefore it follows by induction that  all nodes $y_j^{\kappa},z_j^{\kappa}$ for $j > 2i+1$ and $y_{2i+1}^{\kappa}$ have their natural values. 

Since $g_i^{\kappa}$ is incorrect, all nodes $y_j^{\kappa},z_j^{\kappa}$ for $j \le 2i-1$ have their unnatural values due to Lemma \ref{incorrectnesslemma} and \ref{propagation}. According to Lemma \ref{listeninglemma} and \ref{naturallemma} correcting $g_i^{\kappa}$ does not decrease the cut by a weight of type 3 and gains a weight of type 14. In the following, we distinguish between three cases for the index $i$ and show that $g_{i}^{\kappa}$ is unhappy in each of the cases. First, if $i > 2n+2m$ then there are no node edges of type $5$ or $7$ incident to $g_{i}^{\kappa}$. Thus, $g_{i}^{\kappa}$ is unhappy then. Second, if $2m+1 \le i \leq 2n+2m$ then there are no edges of type $5$ incident to $g_{i}^{\kappa}$. Due to Lemma \ref{resetterlemma} correcting $g_i^{\kappa}$ does not decrease the cut by a weight of type 7. Third, if $i \leq 2m$ then there are no edges of type $7$ incident to $g_{i}^{\kappa}$. Correcting $g_i^{\kappa}$ does not decrease the cut by a weight of type 5 since due to the biases of type $13$ we have $c(u_{i,14})=0$, $c(u_{i,13})=1$ for $i \leq m$ and $c(u_{i,16})=1$, $c(u_{i,15})=0$ for $m < i \leq 2m$. Altogether, $g_{i}^{\kappa}$ is unhappy in each of the three cases which is a contradiction. Thus, $g_{i}^{\kappa}$ is correct for all $i$. Thus, all nodes $y^{\kappa}_i,z^{\kappa}_i$ for $1 \leq i \leq 2N+1$ have their natural values.
\end{proof}

\begin{lemma}\label{dis}
If $c(y_{1}^{\kappa})=c(u^{\kappa})=0$ and $c(u^{\overline{\kappa}})=1$ then $col(d^{\kappa})=1$ and $col(d^{\overline{\kappa}})=0$.
\end{lemma}

\begin{proof}
Assume $c(y_{1}^{\kappa})=c(u^{\kappa})=0$ and $c(u^{\overline{\kappa}})=1$. Then, independently of the color of $y_{1}^{\kappa}$, node $d^{\kappa}$ is biased to $1$ and $d^{\overline{\kappa}}$ to $0$ by type $10$ -- recall that Theorem \ref{substitute} only applies to local optima in which the comparing node is biased to the color that it had if it was a single node. Lemma \ref{propagation} implies $c(y_{0}^{\kappa})=0$. Since $c(u^{\kappa})=0$ and $c(y_{0}^{\kappa})=0$ node $y_{0}^{\kappa}$ and its counterpart, namely the constant $0$, are decisive for $d^{\kappa}$. Thus, Theorem \ref{substitute} implies $col(d^{\kappa})=1$. 

Since $c(u^{\overline{\kappa}})=1$, node $u^{\overline{\kappa}}$ and its counterpart, namely the constant $1$, are decisive for $d^{\overline{\kappa}}$. Thus, Theorem \ref{substitute} implies $col(d^{\overline{\kappa}})=0$.
\end{proof}

\begin{lemma}\label{disumgekehrt}
If $c(y_{1}^{\kappa})=c(u^{\overline{\kappa}})=0$ and $c(y_{1}^{\overline{\kappa}})=c(u^{\kappa})=c(y_{0}^{\overline{\kappa}})=1$ then $col(d^{\overline{\kappa}})=0.$
\end{lemma}

\begin{proof}
Assume $c(y_{1}^{\kappa})=c(u^{\overline{\kappa}})=0$ and $c(y_{1}^{\overline{\kappa}})=c(u^{\kappa})=c(y_{0}^{\overline{\kappa}})=1$. Then, Lemma \ref{propagation} implies $c(z_{0}^{\kappa})=1$ since $c(y_{1}^{\kappa})=0$. Node $d^{\overline{\kappa}}$ is biased to $0$ by type $10$. Since $c(u^{\overline{\kappa}})=0$, $c(y_{0}^{\overline{\kappa}})=1$, and $c(z_{0}^{\kappa})=1$ node $z_{0}^{\kappa}$ and its counterpart, i.\ e. the constant $1$, are decisive for $d^{\overline{\kappa}}$. Thus, Theorem \ref{substitute} implies $c(d^{\overline{\kappa}})=0.$
\end{proof}

\begin{lemma}\label{yzunnat}
If $col(d^{\kappa})=1$, $col(d^{\overline{\kappa}})=0$, and all $y_{i}^{\kappa},z_{i}^{\kappa}$ are biased to their unnatural values by type $8$ then they have their unnatural values.
\end{lemma}

\begin{proof}
Assume that $col(d^{\kappa})=1$, $col(d^{\overline{\kappa}})=0$, and all $y_{i}^{\kappa},z_{i}^{\kappa}$ are biased to their unnatural values by type $8$. Then, $col(d^{\overline{\kappa}})=0$ together with the bias to the unnatural value imply $c(z_{0}^{\kappa})=1$. Then, $col(d^{\overline{\kappa}})=1$ together with the bias to the unnatural value imply $c(y_{0}^{\kappa})=0$. Then, $c(z_{1}^{\kappa})=1$ and therefore $c(y_{1}^{\kappa})=0$. If $c(y_{j-1}^{\kappa})=0$ for any $2 \le j \le 2N+1$ then the bias to the unnatural value implies $c(z_{j}^{\kappa})=1$. Analogously, if $c(z_{j}^{\kappa})=1$ for any $2 \le j \le 2N+1$ then $c(y_{j}^{\kappa})=0$. Thus, the claim follows by induction.
\end{proof}

\begin{observation}\label{yzendearg}
If $c(y_1^{\kappa})=1$ and all $y_{i}^{\kappa},z_{i}^{\kappa}$ are biased to their natural values by type $8$ then $c(z_1^{\kappa})=c(z_0^{\kappa})=0$ and $c(y_0^{\kappa})=1$.
\end{observation}

\begin{lemma}\label{lambdas}
If $col(d^{\overline{\kappa}})=1$, $col(d^{\kappa})=0$ and all nodes $y_{i}^{\overline{\kappa}},z_{i}^{\overline{\kappa}}$ are biased to their unnatural values by type $8$ then $c(x^{\overline{\kappa}})=c(\lambda^{\overline{\kappa}}).$
\end{lemma}

\begin{proof}
Assume $col(d^{\overline{\kappa}})=1$, $col(d^{\kappa})=0$ and that all nodes $y_{i}^{\overline{\kappa}},z_{i}^{\overline{\kappa}}$ are biased to their unnatural values by type $8$. Then, Lemma \ref{yzunnat} implies that $c(y_{2i+1}^{\overline{\kappa}})=0$ for each $0 \le i \le N$. Then, Lemma \ref{resetterlemma} implies $c(\lambda_{i}^{\overline{\kappa}}) \ne c(\mu_{i}^{\overline{\kappa}})$ for all $1 \le i \le n$. Thus, Lemma \ref{listeninglemma} implies that $x^{\overline{\kappa}}_{i}$ takes the colors of $\lambda^{\overline{\kappa}}_{i}$ for any $1 \le i \le n$.
\end{proof}

\begin{lemma}\label{wx0x1}
Assume $c(y_{1}^{\kappa})=col(d^{\overline{\kappa}})=1$ and $col(d^{\kappa})=0$. If $\kappa = 0$ and $C(x^{0}) \ge C(x^{1})$ or $\kappa = 1$ and $C(x^{1}) > C(x^{0})$ then $x^{\overline{\kappa}}=w(x^{\kappa}).$
\end{lemma}

\begin{proof}
Assume $c(y_{1}^{\kappa})=col(d^{\overline{\kappa}})=1$ and $col(d^{\kappa})=0$. We first consider the case that $\kappa = 0$ and $C(x^{0}) \ge C(x^{1})$. Due to $C(x^{0}) \ge C(x^{1})$ all $y_{i}^{0},z_{i}^{0}$ are biased to their natural values by type $8$. Since $c(y_1^{0})=1$ Lemma \ref{incorrectnesslemma} and \ref{propagation} together imply that all gates in $G_{C}^{0}$ compute correctly. Since all gates compute correctly we have $col(w_{i,1}^{0})=col(w_{i,2}^{0})$ for all $1 \le i \le n$. Then, Lemma \ref{resetterlemma} implies $col(w_{i,1}^{0})=c(\lambda_{i}^{1})$ for all $i$ and therefore $w(x^{0})=c(\lambda^{1})$. Moreover, due to the same Lemma we also have $col(w_{i,1}^{0}) \ne c(\mu_{i}^{1})$ for all $i$. Assume for the sake of contradiction, $c(x^{1}) \neq w(x^{0})$. Then all nodes $y_i^{1},z_i^{1}$ are biased to their unnatural values by type $8$. Then, Lemma \ref{yzunnat} implies that they have their unnatural values. Then, Lemma \ref{listeninglemma} implies that a flip of a node $x^{1}_{i}$ for any $i$ does not decrease the cut by a weight of type $3$. Thus the nodes $x^{1}$ assume the colors such that $x^{1}=w(x^{0})$ which is a contradiction. Thus, $x^{1}=w(x^{0}).$ The case $\kappa = 1$ and $C(x^{1}) > C(x^{0})$ is symmetric with the single difference that the case for the equality of $C(x^{0})$ and $C(x^{1})$ is obsolete.
\end{proof}

\begin{lemma}\label{transportToRight}
Assume $c(y_{2i-1}^{\kappa})=1$, and $c(\alpha_{i,j}^{\kappa}) \ne col(I_{j}(g_{i}^{\kappa}))$ for $1 \le j \le 2$. Then, $c(u_{i,13}^{\kappa})=c(u_{i,15}^{\kappa}) \ne col(g_{i}^{\kappa})$ and $c(u_{i,14}^{\kappa})=c(u_{i,16}^{\kappa}) = col(g_{i}^{\kappa})$.
\end{lemma}

\begin{proof}
Assume $c(y_{2i-1}^{\kappa})=1$, and $c(\alpha_{i,j}^{\kappa}) \ne col(I_{j}(g_{i}^{\kappa}))$ for $1 \le j \le 2$. Then, $g_{i}^{\kappa}$ is correct due to Lemma \ref{incorrectnesslemma} and \ref{propagation}. Since $c(y_{2i-1}^{\kappa})=1$, node $u_{i}^{14}$ is biased to $c(\alpha_{i,1}^{\kappa}) \wedge c(\alpha_{i,2}^{\kappa})$ by type $13$. Thus, it is biased to black if and only if $\alpha_{i,1}^{\kappa}$ and $\alpha_{i,2}^{\kappa}$ are both black. But since $c(\alpha_{i,j}^{\kappa}) \ne col(I_{j}(g_{i}^{\kappa}))$ and $col(g_{i}^{\kappa})=\neg(col(I_{1}(g_{i}^{\kappa})) \vee col(I_{2}(g_{i}^{\kappa})))$ node $g_{i}^{\kappa}$ is also black if and only if $\alpha_{i,1}^{\kappa}$ and $\alpha_{i,2}^{\kappa}$ are both black. Thus, $u_{i}^{14}$ is biased to the color that $g_{i}^{\kappa}$ has and $u_{i}^{13}$ to the opposite. Therefore, $c(u_{i,13}^{\kappa}) \ne c(u_{i,14}^{\kappa}) = col(g_{i}^{\kappa})$. 

Since $c(y_{2i-1}^{\kappa})=1$, node $u_{i}^{15}$ is biased to $\neg c(\alpha_{i,1}^{\kappa}) \vee \neg c(\alpha_{i,2}^{\kappa})$ by type $13$. Therefore, it is biased to white if and only if $\alpha_{i,1}^{\kappa}$ and $\alpha_{i,2}^{\kappa}$ are both black. As in the previous case, it follows that $u_{i}^{15}$ is biased to the opposite color that $\hat{g}_{i}^{\kappa}$ has. Since $u_{i}^{16}$ is biased to the opposite color, we have $c(u_{i,15}^{\kappa}) \ne c(u_{i,16}^{\kappa}) = col(\hat{g}_{i}^{\kappa})$ which concludes the proof.
\end{proof}

Now we continue to prove Theorem \ref{plsred}. Let $P$ be a local optimum in $G_C$ and assume w.\ l.\ o.\ g. that $C(x^0) \geq C(x^1)$. Then, all nodes $y_i^{0},z_i^{0}$ are biased to their natural values by type $8$. From Lemma \ref{compareuuneq} we know that $c(u^0) \neq c(u^1)$. In the following, we consider the four possible cases for the vector $c(y_1^0, y_1^1)$ and distinguish within each of these cases between the two cases for $c(u^0,u^{1})$, if necessary. For the three cases for the colors of $c(y_1^0, y_1^1)$ in which at least one node is white we show that they cannot occur in local optima and for the case that both nodes are black we show that the colors of the input nodes $x^{0}$ induce a local optimum of $C$. 

$\underline{(0,0)}$: Due to Lemma \ref{propagation} we have $c(y_0^0)=c(y_0^1)=0$ and $c(z_0^0)=c(z_0^{1})=1$. If $c(u^0,u^1)=(0,1)$ then Lemma \ref{dis} implies $col(d^0)=1$ and $col(d^1)=0$ whereafter Lemma \ref{onezero} implies $c(y_1^0)=1$ which is a contradiction. Now assume $c(u^0,u^1)=(1,0)$. Then, Lemma \ref{dis} implies $col(d^0)=0$ and $col(d^1)=1$. If the nodes $y^{1}_i,z^{1}_i$ are biased to their unnatural values by type $8$ then $c(x^{1})=c(\lambda^{1})$ due to Lemma \ref{lambdas}. Then, at least one of the conditions $w(x^{0}) \ne c(x^{1})$, and $w(x^{0}) = c(\lambda^{1})$ is violated. But this implies that type $8$ biases the nodes $y^{1}_i,z^{1}_i$ to their natural value which is a contradiction. Therefore, all nodes $y_i^{1},z_i^{1}$ are biased to their natural values by type $8$. But then Lemma \ref{onezero} implies that $c(y_1^1)=1$ which is also a contradiction.

$\underline{(0,1)}$: According to Lemma \ref{propagation} we have $c(y_0^0)=0$ and $c(z_0^{0})=1$. If $c(u^0,u^1)=(0,1)$ then Lemma \ref{dis} implies $col(d^0)=1$ and $col(d^1)=0$ whereafter Lemma \ref{onezero} implies $c(y_1^0)=1$ which is a contradiction. Now assume $c(u^0,u^1)=(1,0)$. Assume for the sake of contradiction that $y_0^1$ is white. If the nodes $y_{i}^{1},z_{i}^{1}$ are biased to their natural values by type $8$ then Observation \ref{yzendearg} implies $c(y_0^1)=1$ which is a contradiction. If they are biased to their unnatural values then $c(z_{1}^{1})=1$ whereafter $c(y_{1}^{1})=0$ which is also a contradiction. Thus, $y_0^1$ is black. Then, Lemma \ref{disumgekehrt} implies $col(d^1)=0$, but then $c(u^1)=1$ due to the bias of type $9$ which is again a contradiction.

$\underline{(1,0)}$: Lemma \ref{propagation} and Observation \ref{yzendearg} together imply $c(y_0^0)=c(z_0^1)=1$ and $c(z_0^0)=c(y_0^{1})=0$. If $c(u^0,u^1)=(0,1)$ then Lemma \ref{disumgekehrt} implies $col(d^0)=0$. But then $c(u^0)=1$, since $u^0$ is biased by type $9$ to the color of $y_1^{0}$, i.\ e. to 1, which is a contradiction. Now assume $c(u^0,u^1)=(1,0)$. Then, $col(d^0)=0$ and $col(d^1)=1$ due to Lemma \ref{dis}. Due to Lemma \ref{wx0x1} we have $x^1=w(x^0)$ in which case all nodes $y_i^1,z_i^1$ are biased to their natural values by type $8$. Then, Lemma \ref{onezero} implies $c(y^1_{1})=1$ which is a contradiction.

$\underline{(1,1)}$: Due to Observation \ref{yzendearg} we have $c(y_0^0)=1$ and $c(z_0^{0})=0$. Since $c(y_1^0)=c(y_1^1)=1$, Lemma \ref{incorrectnesslemma} and Lemma \ref{propagation} together imply that all gates in $G_C^0$ and $G_C^1$ are correct and $y^{\kappa}_{i},z^{\kappa}_{i}$ have their natural values for $\kappa \in \{0,1\}, 2 \le i \le 2N+1$. Claim \ref{beta12correct} implies $c(\alpha_{j,k}^{0}) \ne col(I_{k}(g_{j}^{0}))$ and $c(\alpha_{i,j}^{1}) \ne col(I_{j}(g_{i}^{1}))$ for all $i$ and $1 \le j \le 2$. Moreover, we also have $col(g_{i}^{\kappa}) \ne col(\hat{g}_{i}^{\kappa})$ for all $i,\kappa$. According to type 10, node $d^{0}$ is biased to $0$ and node $d^{1}$ to $1$. 

In the following, we first show $col(d^{0})=0$ by naming the decisive neighbors of $d^{0}$ and showing that their color is black. For this, we distinguish four cases. First, if $c(u^{0})=1$ then $u^{0}$ and its counterpart, i.\ e. the constant $1$, are decisive for $d^{0}$. Second, if $c(u^{0})=0$ and $c(z_{0}^{1})=1$ then $z_{0}^{1}$ and its counterpart, also the constant $1$, are decisive since $y_0^0$ is black and its counterpart is a white constant. Third, if $c(u^{0})=0$, $c(z_{0}^{1})=0$, and $C(x^0) > C(x^1)$ then the pair of nodes $g_i^0,\hat{g}_i^1$ with highest index $i$ for $i \leq m$ such that $g_i^0$ and $\hat{g}_i^1$ have the same color are both black. Then, Lemma \ref{transportToRight} implies $c(u_{i,13}^{0})=c(u_{2m+i,15}^{1})=0$ and $c(u_{i,14}^{0})=c(u_{2m+i,16}^{1})=1$ -- recall that $c(y_{2i-1}^{\kappa})=1$ for $\kappa \in \{0,1\}$. Moreover, since $g_{j}^{0} \ne g_{2m+j}^{1}$ for all $j>i$ the same Lemma implies $c(u_{j,13}^{0}) \ne c(u_{j,14}^{0}) = g_{j}^{0} \ne g_{2m+j}^{1} = c(u_{2m+j,16}^{1}) \ne c(u_{2m+j,15}^{1})$ for all $j>i$. Then, the nodes $u_{i,14}^{0}$ and $u_{2m+i,16}^{0}$ are decisive for $d^{0}$. Fourth, if $c(u^{0})=0$, $c(z_{0}^{1})=0$, and $C(x^0) = C(x^1)$ then the neighbors of type $5$ representing the constant $1$ adjacent to $d^{0}$ via edges of relative weight $1$ are decisive for $d^{0}$. In any of the above cases the decisive neighbors of $d^{0}$ are black. Since node $d^0$ is biased to $0$ by type $10$ Theorem \ref{substitute} implies $col(d^0)=0$. Then, the bias of type $9$ implies $c(u^0,u^1)=(1,0)$. 

Now we show that $col(d^{1})=1$. For this, we distinguish three cases. First, if $c(y_{0}^{1})=0$ then $y_{0}^{1}$ and its counterpart, the constant $0$, are decisive since $c(u^{1})=0$. Second, if $c(y_{0}^{1})=1$ and $C(x^0) > C(x^1)$ then the pair of nodes $g_i^1,\hat{g}_i^0$ with highest index $i$ for $i \leq m$ such that $g_i^1$ and $\hat{g}_i^0$ have the same color are both white. Then, Lemma \ref{transportToRight} implies  $c(u_{i,13}^{1})=c(u_{2m+i,15}^{0})=1$, $c(u_{i,14}^{1})=c(u_{2m+i,16}^{0})=0$. Moreover, the same Lemma implies $c(u_{j,13}^{1}) \ne c(u_{i,14}^{1}) = g_{j}^{0} \ne g_{2m+j}^{1} = c(u_{2m+j,16}^{0}) \ne c(u_{2m+j,15}^{0})$ for all $j>i$. Then, the nodes $u_{i,14}^{1}$ and $u_{2m+i,16}^{1}$ are decisive for $d^{1}$. Third, $c(y_{0}^{1})=1$, and $C(x^0) = C(x^1)$ then the neighbors of type $5$ representing the constant $0$ adjacent to $d^{1}$ via edges of relative weight $1$ are decisive for $d^{1}$. In any of the above cases the decisive neighbors of $d^{1}$ are white. Since node $d^1$ is biased to $0$ by type $10$ Theorem \ref{substitute} implies $col(d^1)=1$.

According to Lemma \ref{wx0x1} we obtain $x^1=w(x^0)$. Due to our assumption that $C$ only returns its input as better neighbor if the input is locally optimal and since $x^1=w(x^0)$ the colors of $x^0$ induce a local optimum of $C$. This finishes the proof of Theorem \ref{plsred}.
\end{proof}		

\section{Smoothed Complexity of Local Max-Cut}
\label{secsmoothed}

We consider the smoothed complexity of local \mc\ 
for graphs with degree  $O(\log n)$. Smoothed analysis, as 
introduced by Spielman and Teng \cite{ST04}, is motivated by 
the observation that practical data is often subject to some small 
random noise. 
Formally, let $\Omega_{n,m}$ be the set of all weighted graphs 
with $n$ vertices and $m$ edges, in which each graph has maximum degree 
$O(\log n)$. 
In this paper, if $A$ is an algorithm on graphs with 
maximum degree $O(\log n)$, 
then the smoothed complexity of $A$ with $\sigma$-Gaussian perturbation 
is 
$$ Smoothed_A^{\sigma} (n) =  \max_{m} \max_{G \in \Omega_{n,m}} \E_{\mbox{x}_m}[T_A(G^{w_{\max} \cdot \mbox{x}_m})] ,$$
where $\mbox{x}_m =(x_1, \dots , x_m)$ 
is a vector of length $m$, in which each entry is 
an independent Gaussian random variable of standard deviation $\sigma$ and 
mean $0$. $\E_{\mbox{x}_m}$ indicates that the expectation is taken 
over vectors $\mbox{x}_m$ according to the 
distribution described before, 
$T_A(G)$ is the running time of $A$ on $G$, and $G^{w_{\max} \cdot 
\mbox{x}_m}$ 
is the graph obtained from $G$ by adding $w_{\max} \cdot x_i$
to the weight of the $i$-th edge in $G$, where $w_{\max}$ is the 
largest weight in the graph. We assume that the 
edges are considered according to some arbitrary but fixed ordering.  

According to Spielman and Teng, 
an algorithm $A$ has polynomial smoothed complexity 
if there exist positive constants $c'$, $n_0$, $\sigma_0$, $k_1$, and $k_2$
such that for all $n > n_0$ and $0 \leq \sigma< \sigma_0$ we have 
$$ Smoothed_A^{\sigma} (n) < c' n^{k_1} \cdot \sigma^{-k_2} .$$
In this paper, we use a relaxation of polynomial 
smoothed complexity \cite{ST09}, which builds up on Blum and Dungan \cite{BD02}  
(see also Beier and V\"ocking \cite{BV04}). According to this relaxation,
an algorithm $A$ has probably polynomial smoothed complexity 
if there exist positive constants $c'$, $n_0$, $\sigma_0$, and $\alpha$
such that for all $n > n_0$ and $0 \leq \sigma< \sigma_0$ we have
$$\max_{m} \max_{G \in \Omega_{n,m}} 
\E_{\mbox{x}_m}[T^{\alpha}_A(G^{w_{\max} \cdot \mbox{x}_m})]
< \frac{c' n}{\sigma} .$$ 

\begin{theorem}
\label{thesmoothed}
Let $A$ be some FLIP local search algorithm for local \textsc{Max-Cut}. 
Then, $A$ has probably polynomial smoothed complexity on any graph with 
maximum degree $O(\log n)$.
\end{theorem}
\begin{proof}
Let $V=\{ v_1, \dots , v_n\}$, and denote 
by $d_i$ the degree of $v_i$. 
Furthermore, let $w_{i,j}$ be the 
weight of egde $(v_i,v_j)$. 
Let $m = |E|$, and $\mbox{x}_m =
(x_1, \dots , x_m)$ a vector of Gaussian random variables 
of standard deviation $\sigma$ and mean $0$. Alternatively, 
we denote by $x_{i,j}$ the Gaussian random variable which perturbates 
edge $(v_i,v_j)$, i.\ e., $\tilde{w}_{i,j} = w_{i,j} + w_{\max} 
\cdot x_{i,j}$ represents the weight of $(v_i,v_j)$ in the perturbated 
graph $G^{w_{\max} \cdot \mbox{x}_m}$. 

In the following, $G$ is an arbitrary graph in $\Omega_{n,m}$ where $m=O(n \log n)$. We show that for any $\delta \in (0,1)$ there are constants 
$c'$, $n_0$, $\sigma_0$, $k_1$, and $k_2$
such that for all $n > n_0$ and $0 \leq \sigma< \sigma_0$ we obtain
\begin{eqnarray}
\label{eq1smoothed}
 Pr_{\mbox{x}_m}[T_A(G^{w_{\max} \cdot \mbox{x}_m}) 
< \delta^{-2} c' n^{k_1} \cdot \sigma^{-k_2} ] > 1- \delta .
\end{eqnarray}
Then, (\ref{eq1smoothed}) implies the statement of the 
theorem (cf.~\cite{BV04}).

In order to show the inequality above, we make use of the fact that 
the sum of $k$ Gaussian random variables with variance $\sigma^2$ and 
mean $0$ is a Gaussian random variable with variance $k\sigma^2$ 
and mean $0$. Let $X_1, \dots , X_k$ be $k$ Gaussian random variables
with variance $\sigma^2$ and mean $0$. Furthermore, let $a$ be some 
real number, and $S \subset \{ 1 , \dots , k\}$. 
Then, we can state the following claim.
\begin{claim}
\label{claim1smoothed}
For some large constant $c$ and any $\delta' \in (0,1)$ 

\begin{equation}\label{equationsmoothed} Pr \left[ \left| \sum_{j \in S} X_j - \sum_{j \not\in S} X_j-a \right| \leq 
\frac{\delta' \sigma}{c\cdot 2^k}\right] \leq \delta' \cdot 2^{-k} .
\end{equation}
\end{claim}

\begin{proof}
Let $X= \sum_{j \in S} X_j$ and $Y=\sum_{j \not\in S} X_j$. Then,
$X$ is a Gaussian random variable with variance $|S|\sigma^2$, and 
$Y$ is a Gaussian random variable with variance $(k-|S|)\sigma^2$. 
Since $X$ is distributed according to the density function 
$\frac{1}{\sqrt{2\pi|S|}\sigma}e^{-x^2/(|S|\sigma)^2}$, for any real
$b$ it holds that 
\begin{eqnarray*}
Pr\left[ |X-b| \leq \frac{\delta' \sigma}{c\cdot 2^k}\right]
&\leq& \int_{b-\frac{\delta' \sigma}{c\cdot 2^k}}^{b+\frac{\delta' 
\sigma}{c\cdot 2^k}} \frac{1}{\sqrt{2\pi|S|}\sigma}e^{-x^2/(|S|\sigma^2)} 
\mbox{d} x \\
&\leq& \frac{1}{\sqrt{2\pi|S|}\sigma} \cdot 
\frac{2 \delta' \sigma}{c\cdot 2^k} \leq \delta' \cdot 2^{-k}
\end{eqnarray*}
Setting $b=Y+a$ we obtain the claim.
\end{proof}

In order to show the theorem, we normalize the weights by setting 
the largest weight to $1$, and dividing all other weights by 
$w_{\max}$. That is, we obtain some graph $G'$ with weights 
$w'_{i,j}= w_{i,j}/w_{\max}$. The edge weights of $G'$ are 
perturbated accordingly by Gaussian random variables with variance 
$\sigma^2$ and mean $0$. Clearly, $T_A(G') =T_A(G)$ and 
$Pr[T_A(G'^{\mbox{x}_m}) 
< \delta^{-1} c' n^{k_1} \cdot \sigma^{-k_2} ] =
Pr[T_A(G^{w_{\max} \cdot \mbox{x}_m}) 
< \delta^{-1} c' n^{k_1} \cdot \sigma^{-k_2} ] $. Therefore, we 
consider $G'$ instead of $G$ in the rest of the proof.

In the next step, we show that for an arbitrary but fixed partition $P$ of $G'$
and node $v_i$, flipping $v_i$ 
increases (or decreases) the cut by $\Omega \left( \frac{\delta \sigma}{n 2^{d_i}}
\right)$, with probability $1-\delta/2 \cdot  
n^{-1} 2^{-d_i}$. This is easily obtained 
from Claim (\ref{claim1smoothed})
in the following way. Define 
$S'$ to be the set of the neighbors of $v_i$, which are in the same partition 
as $v_i$ according to $P$. Let $e_1, \dots , e_{d_i}$ be the edges incident to 
$v_i$, and denote by $w_1, \dots , w_{d_i}$ the weights of 
these edges in $G'$. We assume w.\ l.\ o.\ g.~that $e_1, \dots , e_{|S'|}$ 
have both ends in the same partition as $v_i$, and $S = \{ 1, \dots , |S'|\}$.
Furthermore, let $a=\sum_{j \not\in S} w_j - \sum_{j \in S} w_j$, $k=d_i$,
and $\delta' = \delta/(2n)$. Applying now Equation (\ref{equationsmoothed})
we obtain the desired result.

For a node $v_i$, 
there are at most $\sum_{i=0}^{d_i} \left({{d_i} \atop i } \right)
= 2^{d_i}$ possibilities to partition the 
edges into two parts, one subset in the same partition as $v_i$
and the other subset in the other partition. Therefore, 
by applying the union bound we conclude that any flip of an unhappy $v_i$  
increases the cut by 
$\Omega \left( \frac{\delta \sigma}{n 2^{d_i}} \right)$,
with probability at least $1-\delta/2 \cdot n^{-1}$. 
Since there are $n$ nodes in total, 
we may apply the union bound again, and obtain that every flip (carried out by some unhappy 
node) increases 
the cut by $\Omega \left( \frac{\delta \sigma}{n 2^{d_i}} \right)$,
with probability at least $1-\delta/2$. 
Since 
$d_i = O(\log n)$ and 
the largest weight in $G'$ is $1$, we conclude that the 
largest cut in $G'$ may have weight $O(n \log n)$. Furthermore, 
for each $i$ we have $|x_i| \leq l \sqrt{\ln n}$ with probability 
$1-O(n^{-l})$ whenever $l$ is large enough (remember that $\sigma<1$). 
Let $A_1$ be the event that there is some $x_i$ with $|x_i| = \omega(\log n)$, 
and $A_2$ is the event that there is a node $v_i$ and a partition $P$ such that 
flipping $v_i$ increases the cut by at most  
$\tau \left( \frac{\delta \sigma}{n 2^{d_i}}
\right)$, where $\tau$ is a very small constant. 
We know that $Pr[A_1] = n^{-\omega(1)}$ and 
$Pr[A_2] < \delta/2$.
Thus,
as long as $\delta = n^{-O(1)}$, the total number of 
steps needed by $A$ is at most 
$$T_A(G'^{\mbox{x}_m}) = 
O \left( \frac{n^2 \log^2 n 2^{d_i}}{\delta \sigma} \right) 
= \frac{n^{O(1)}}{\delta \sigma}$$
with probability $1-(Pr[A_1] + Pr[A_2]) > 1-\delta$. 

Now we consider the case when $\delta = n^{-\Omega(1)}$. Again, 
let $A_1$ be the event that there is some $x_i$ with $|x_i| = \omega(\log \delta^{-1})$. 
Since $x_i$ is a Gaussian random variable, $Pr[A_1] = \delta^{\omega(1)}$. 
On the other hand, let $A_2$ be the event that 
there is a node $v_i$ and a partition $P$ such that 
flipping $v_i$ increases the cut by at most 
$\tau \left( \frac{\delta \sigma}{n 2^{d_i}}
\right)$, where $\tau$ is a very small constant. Again, 
$Pr[A_2] < \delta/2$. Then, 
the total number of 
steps needed by $A$ is at most 
$$T_A(G'^{\mbox{x}_m}) = 
O \left( \frac{n^2 \log n \cdot \log \delta^{-1} 
\cdot 2^{d_i}}{\delta \sigma} \right) 
= \frac{n^{O(1)} \log \delta^{-1}}{\delta \sigma}$$
with probability $1-(Pr[A_1] + Pr[A_2]) > 1-\delta$. Note that if 
$\max_i |x_i| = \log \delta$ then the input size $|G'^{\mbox{x}_m}|$ is 
$O(n \log \log \delta)$. Hence, the result above does not imply that 
$T_A(G'^{\mbox{x}_m}) = \frac{|G'^{\mbox{x}_m}|^{O(1)}}{\delta \sigma}$ 
with probability $1-\delta$. 
\end{proof}

\section{Conclusion and Open Problems}
In this paper, we introduced a technique by which we can substitute graphs with certain nodes of unbounded degree, namely so called comparing nodes, 
by graphs with nodes of maximum degree five such that local optima of the former graphs induce unique local optima of the latter ones. 
Using this technique, we show that the problem of computing a local optimum of the \mc\ problem is \pls-complete even on graphs with maximum degree five. 
We do not show that our \pls-reduction is tight, but the tightness of our reduction would not result in the typical knowledge gain anyway since the properties that 
come along with the tightness of \pls-reductions, namely the \pspace-completeness of the standard algorithm problem and the existence of instances that are 
exponentially many improving steps away from any local optimum, are already known for the maximum degree four \cite{monien1}. The obvious remaining question 
is to ask for the complexity of local \mc\ on graphs with maximum degree four. Is it in \Pol? Is it \pls-complete? Another important question 
is whether local \mc\ has in general probably polynomial smoothed complexity. Unfortunately, the methods used so far seem not to be
applicable to show that in graphs with super-logarithmic degree
the local \mc\ problem has probably polynomial smoothed complexity (cf.~also \cite{heiko1}).
\paragraph{\textbf{Acknowledgement.}}
We thank Dominic Dumrauf, Martin Gairing, Martina H\"ullmann, Burkhard Monien, and Rahul Savani for helpful suggestions.

\newpage

\appendix

\section{Appendix}
\subsection{Definition of the nodes of type I and III and proof of Lemma \ref{Gf}}

We introduce the types I and III of the nodes of degree $\leq 4$, since their use simplifies the proof.

\begin{definition}
For a node $u$ and edges $a_{u},b_{u},c_{u},d_{u}$ incident to $u$ with $w(a_{u}) \ge w(b_{u}) \ge w(c_{u}) \ge w(d_{u})$ we distinguish the following types for $u$: 

	\begin{itemize}
		\item type I: if $w(a_{u}) > w(b_{u}) + w(c_{u}) + w(d_{u})$
		\item type III: if $w(a_{u}) + w(d_{u}) < w(b_{u}) + w(c_{u})$
	\end{itemize}

\end{definition}

\begin{observation}\label{basics}
		For a graph $G=(V,E)$, a partition $P$, a node $u$, and edges $a_{u},b_{u},c_{u},d_{u}$ incident to $u$ with $w(a_{u}) \ge w(b_{u}) \ge w(c_{u}) \ge w(d_{u})$ the following two conditions hold:
		
		\begin{itemize}
			\item If $u$ is of type I then $u$ is happy in $G_{P}$ if and only if $a_{u}$ is in the cut. 
			\item If $u$ is of type III then $u$ is happy in $G_{P}$ if and only if at least two of the edges $a_{u}$, $b_{u}$,$c_{u}$ are in the cut.
		\end{itemize}
	\end{observation}
	
The original statement of Lemma \ref{Gf} in \cite{monien1} contained the following two properties.

\begin{theorem}\label{P-hardness}
\begin{enumerate}[label=(\alph*)]
\item \label{P-hard} Finding a local optimum of \mc\ on graphs containing only nodes of type I and III is $P$-hard with respect to logspace-reduction.
\item 
\label{Gfs} Let $f:\{0,1\}^{n} \rightarrow \{0,1\}^{m}$ be a function and $C$ be a boolean circuit with $N$ gates which computes $f$. Then, using $O(log N)$ space, one can compute a graph $G^{C}=(V^{C},E^{C})$ containing only nodes of type I and III among which there are nodes $s_{1},\ldots,s_{n},t_{1},\ldots,t_{m} \in V^{C}$ of degree one such that $f(c_{P}(s))=c_{P}(t)$ in every local optimum $P$ of $G^{C}$.
\end{enumerate}
\end{theorem}

\begin{proof} At first, we show property \ref{P-hard}. We reduce from the P-complete problem circuit-value \cite{ladner}. An instance of the circuit-value problem is a Boolean circuit $C$ together with an assignment for the inputs of $C$. The problem asks for the output of $C$ on the given input assignment. For the gates, we assume w.\ l.\ o.\ g. that they are either NOR-gates with a fanin of two and fanout of one or NOT-gates with a fanin of one and fanout of at most two. We order the gates $g_{N},g_{N-1},\ldots,g_{1}$ topologically such that if $g_{i}$ is an input of $g_{j}$ then $i>j$. We let $g_{m},\ldots,g_{1}$ be the output of the $C$ and assume w.\ l.\ o.\ g.  that the gates $g_{N},\ldots,g_{N-n+1}$ are NOT-gates, that they are the only gates in which the inputs of the circuit $C$ occur, and that the output gates are NOT-gates with fanout zero. Let $I_{1}(g_{i})$ and $I_{2}(g_{i})$ be the gates which are the inputs of a NOR-gate $g_{i}$ for $1 \le i \le N-n$, $I(g_{i})$ be the input of a NOT-gate $g_{i}$ for $1 \le i \le N-n$, and $value(g_{i})$ be the corresponding value of the assignment of the input of $g_{i}$ with $N-n+1 \leq i \leq N$.

We construct a graph $T=(V,E)$ with weights $w:E \rightarrow \mathbb{N}$ from $C$ in the following way. The set of nodes is $V=\{v_1,\ldots,v_{3N+1}\}$ and the set $E$ as follows:

\begin{enumerate}[label=(\roman*)]
\item \label{erste} For each NOR-gate $g_i$ with $1 \le i \le N-n$ we have the following edges of weight $2^i$: $\{v_i,v_j\} \in E$ if $I_k(g_i)=g_j$ for $1 \le k \le 2, 1 \le j \le N$, and $\{v_i,v_{N+2i}\} \in E$. 
\item \label{zweite} For each NOT-gate $g_i$ with $1 \le i \le N-n$ with $I(g_i)=g_j$ for $1 \le j \le N$ we have $\{v_i,v_j\} \in E$ with $w(\{v_i,v_j\}) = 2^i$. 
\item \label{dritte} For each $g_i$ with $N-n+1 \le i \le N$ we have the following edge of weight $2^i$: if $value(g_i)=1$ then $\{v_i,v_{N+2i}\} \in E$, otherwise $\{v_i,v_{N+2i-1}\} \in E$.
\item \label{vierte} $\{v_i,v_{i+1}\} \in E$ for all $N+1 \le i \le 3N$ with a weight of $2^i$. 
\end{enumerate}

Then, each node $v_i$ for $N+1 \le i \le 3N$ has a degree of at most three, is of type I, and the heaviest edge incident to $v_i$ is $\{v_i,v_{i+1}\}$. The node $v_{3N+1}$ has a degree of one, is also of type I, and its heaviest edge is $\{v_{3N+1},v_{3N}\}$. Each node $v_i$ for $N-n+1 \le i \le N$ has a degree of two, is of type I, and its heaviest edge is the one described in \ref{dritte}. Each node $v_i$ for $1 \le i \le N-n$, for which $g_i$ is a NOT-gate, has a degree of at most three, is of type I, and its heaviest edge is the edge of weight $2^i$ described in \ref{zweite}. However, if $g_i$ is a NOR-gate then $v_i$ has a degree of four, is of type III, and the edges with influence on it are the three edges of weight $2^i$ described in \ref{erste}.

Consider a local optimum $P$ of $T$. Due to the symmetry of the local \mc\ problem we may assume w.\ l.\ o.\ g.  that $c_{P}(v_{3N})=1$. In the following, we use Observation \ref{basics} to derive the colors of the remaining nodes as induced by their types. First, $c_{P}(v_{i}) \ne c_{P}(v_{i+1})$ for all $N+1 \le i \le 3N$. Thus, $c_{P}(v_{N+2i}) = 1$ and $c_{P}(v_{N+2i-1}) = 0$ for all $1 \le i \le N$ and $c_{P}(v_{3N+1}) = 0$. Then, for each $N-n+1 \le i \le N$ we have $c_{P}(v_{i}) = 1$ if $value(g_i)=1$ and $c_{P}(v_{i}) = 0$ otherwise, i.\ e. the colors of the nodes $v_i$ for $N-n+1 \le i \le N$ correspond to the assignment for the inputs of $C$. Now consider the nodes $v_i$ for $1 \le i \le N-n$. If $g_i$ is a NOT-gate with $I(g_i)=g_j$ for $m+1 \le j \le N$ then $c_P(v_i) \ne c_P(v_j)$, i.\ e. the color of $v_i$ corresponds to the output of a NOT-gate w.\ r.\ t. the color of $v_j$. Finally, if $g_i$ is a NOR-gate with $I_1(g_i)=g_k$ and $I_2(g_i)=g_j$ for $m+1 \le j < k \le N$ then $c_P(v_i)=1$ if and only $c_P(v_j)=c_P(v_k)=0$ since $v_i$ is of type III and its neighbor $v_{N+2i}$ is already known to be black. Thus, the color of $v_i$ corresponds to the output of a NOR-gate with respect of the colors of $v_j$ and $v_k$. Therefore, the color of each node $v_i$ for $1 \le i \le N-n$ corresponds to the output of $g_i$ in $C$. In particular, the colors of $v_1,\ldots,v_m$ correspond to the output of $C$.

In the following we show that our reduction is in logspace. The number of edges in $T$ is linear in N since each node has maximum degree four. The weights of the edges are powers of two. Thus, we only need to store the exponents of the weights. If we write an edge weight to the output tape then we first write the ``1'' for the most significant bit and then we write as often a ``0'' as determined by the exponent.

Now we show \ref{Gfs}. We let $G^{C}=(V^{C},E^{C})$ be the graph obtained from $T$ by omitting the edges described in \ref{dritte}. Furthermore, we let $s_i=v_{N-n+i}$ for $1 \leq i \leq n$ and $t_{j}=v_{j}$ for $1 \leq j \leq m$. Then, the nodes $s_{i}$ and $t_{j}$ are of degree one. As in the proof for \ref{P-hard} we get $f(c_{P}(s))=c_{P}(t)$.
\end{proof}

\end{document}